\documentclass{article}

\usepackage{geometry}
\usepackage[utf8]{inputenc}
\usepackage{babel}
\usepackage{amsmath,amsfonts,amssymb,amsthm,yhmath,mathrsfs, amsbsy}
\usepackage{latexsym}
\usepackage{mathtools}
\usepackage{fancyhdr}
\usepackage[nottoc]{tocbibind}
\usepackage{indentfirst}
\usepackage{paralist}
\usepackage[a-1b]{pdfx}
\usepackage[pdfa]{hyperref}
\usepackage{graphicx}
\usepackage{dcpic, pictex}
\usepackage{etoolbox}
\usepackage{tikz}
\usepackage{enumerate}
\usepackage{makeidx}
\usepackage{multicol}
\usepackage{comment}

\newtheorem{theorem}{Theorem}[section]
\newtheorem{proposition}[theorem]{Proposition}
\newtheorem{lemma}[theorem]{Lemma}
\newtheorem{corollary}[theorem]{Corollary}

\theoremstyle{definition}
\newtheorem{definition}[theorem]{Definition}

\newtheorem{remark}[theorem]{Remark}
\newtheorem{example}[theorem]{Example}

\newcommand{\N}{\mathbb N}

\newcommand{\Z}{\mathbb Z}

\newcommand{\F}{\mathbb F}
\newcommand{\Cc}{\mathcal C}

\newcommand{\MM}{\mathbb M}
\DeclareMathOperator{\rk}{rk}
\DeclareMathOperator{\srk}{srk}

\newcommand{\Cd}{\mathcal{C}^\perp}
\newcommand{\A}{\mathcal{A}}

\newcommand{\matsumi}{\mathbb M}
\newcommand{\wt}{\mathrm{wt}}
\newcommand{\gauss}{\genfrac{[}{]}{0pt}{}}
\DeclareMathOperator{\maxrk}{maxrk}
\DeclareMathOperator{\maxsrk}{maxsrk}

\DeclareMathOperator{\rowsp}{rowsp}
\DeclareMathOperator{\colsp}{colsp}
\DeclareMathOperator{\supp}{supp}

\title{Sum-rank metric codes}
\author{Elisa Gorla, Umberto~Mart\'inez-Pe\~nas, Flavio Salizzoni}
\date{}

\begin{document}

\maketitle

\begin{abstract}
Sum-rank metric codes are a natural extension of both linear block codes and rank-metric codes. They have several applications in information theory, including multishot network coding and distributed storage systems. The aim of this chapter is to present the mathematical theory of sum-rank metric codes, paying special attention to the $\F_q$-linear case in which different sizes of matrices are allowed. We provide a comprehensive overview of the main results in the area. In particular, we discuss invariants, optimal anticodes, and MSRD codes. In the last section, we concentrate on $\F_{q^m}$-linear codes.
\end{abstract}
\tableofcontents

\section{Definitions and notation}\label{sect:defn}

For a positive integer $r$, we denote by $[r]$ the set $\{1,\ldots,r\}$. For a prime power $q$ and positive integers $m \geq n$, let $\F_q^{m \times n}$ be the set of $m \times n $ matrices with entries in the finite field $\F_q$. We denote by $\rk(M)$ the rank of a matrix $M \in \F_q^{m \times n}$. Throughout this chapter we denote by $\MM$ the $\F_q$-linear vector space 
 $$\MM=\F_q^{m_1\times n_1}\times\ldots\times\F_q^{m_\ell\times n_\ell},$$
where $\ell ,m_1, \ldots ,  m_\ell,n_1, \ldots , n_\ell$ are positive integers. In order to simplify some statements and improve the readability of the text, we assume that $m_1\geq  \ldots \geq   m_\ell$ and 
$ n_i\leq m_i$ for $i\in[\ell]$. Moreover let $n = n_1 + \ldots + n_\ell$ and if $m_1=\dots=m_{\ell}$ we write $m$ in place of $m_i$.
An element $C\in\MM$ is called codeword and it can be written as $C = (C_1, \ldots , C_\ell)$, with $C_i\in\F_q^{m_i\times n_i}$ for $i\in[\ell]$. The sum-rank weight of $C$ is given by
$$\mathrm{srk}(C) = \sum_{i=1}^\ell \rk(C_i).$$
This weight naturally induces a metric on $\MM$, called the sum-rank metric. Indeed, let $d$ be the map 
$$\begin{array}{ccccc}
d  &: & \MM \times \MM & \longrightarrow & \N \\
& & (C, D) & \longmapsto & \mathrm{srk}(C-D),
\end{array}$$
then $(\MM,d)$ is a metric space.
\begin{definition}
A linear sum-rank metric code $\Cc$ is an $\F_q$-linear subspace of $(\MM,d)$. 
\end{definition}
Throughout the chapter, we often refer to it simply as a code if the metric is clear from the context. We say that a code $\Cc\subseteq\MM$ is non-trivial if $\Cc\neq 0,\MM$.
\begin{remark}
The sum-rank metric is a natural generalization of both the Hamming metric and the rank metric. Indeed, on the one side if $\ell = 1$, then $\MM= \F_q^{m \times n}$ and the sum-rank metric coincides with the rank metric on $\MM$. In this case a sum-rank metric code $\Cc$ is a rank-metric code. 
On the other side if $m_1 = \ldots = m_\ell = 1 $, then $\MM=\F_q^n$ and the sum-rank metric coincides with the Hamming metric. In this situation, a sum-rank metric code $\Cc$ is a block code endowed with the Hamming metric. 

As a consequence, a result proved in the sum-rank metric case holds true also for the rank and the Hamming metric. The mathematical theory of sum-rank metric codes, however, tends to be more complex than that of rank-metric and linear block codes. E.g., there are results that hold for both rank-metric and linear block codes, but do not hold in full generality for sum-rank metric codes. An example of such a result is Wei duality, which we discuss in the next section.
\end{remark}
\begin{remark}
It is worth noting that every sum-rank metric code can be viewed as a code with the rank metric in the appropriate ambient space. For instance, consider the space $\MM$ and let $\bar m=m_1+\dots+m_{\ell}$. We denote by $X(\MM)$ the set
\begin{equation*}
\begin{split}
    X(\MM)=\{(s,t)\in[\bar m]\times[n]:&\text{ if }m_1+\dots+m_i<s\leq m_1+\dots+m_{i+1}\\
    &\text{ then }n_1+\dots+n_i<t\leq n_1+\dots+n_{i+1}\}.
    \end{split}
\end{equation*}
Let $\F_q^{\bar m\times n}[\MM]$ be the $\F_q$-linear space of matrices supported on $X(\MM)$, i.e.,
\begin{equation*}
    \F_q^{\bar m\times n}[\MM]=\{M\in\F_q^{\bar m\times n}: M_{s,t}\neq0 \mbox{ only if } (s,t)\in X(\MM)\}.
\end{equation*}
Since $\F_q^{\bar m\times n}[\MM]$ is a subspace of $\F_q^{\bar m\times n}$, we can equip it with the metric induced by the rank metric on $\F_q^{\bar m\times n}$. Then, the $\F_q$-linear isometry $\iota:\MM\rightarrow\F_q^{\bar m\times n}[\MM]$ given by
\begin{equation*}
    (C_1,\dots,C_{\ell})\longmapsto\begin{pmatrix} C_1&&\\&\ddots&\\&&C_{\ell}\end{pmatrix}
\end{equation*}
is distance-preserving, i.e., $\srk(C)=\rk(\iota(C))$ for all $C\in\MM$. Therefore, a sum-rank metric code $\Cc\subseteq\MM$ can be always identified with its image $\iota(\Cc)$, that is a rank-metric space in $\F_q^{\bar m\times n}[\MM]$.
\end{remark}
Now we define two parameters that play a fundamental role in the theory of sum-rank metric codes.
\begin{definition}
The minimum distance of a code $0\neq\mathcal{C} \subseteq \MM$ is 
$$ d(\Cc) = \min\{ \mathrm{srk}(C) : C \in \Cc \setminus \{ 0\}  \}$$ and the maximum sum-rank distance is  
$$\maxsrk(\Cc) = \max\{ \mathrm{srk}(C) : C \in \Cc  \}.$$
\end{definition}
We say that a codeword $C\in\Cc$ realizes the minimum distance (respectively, the maximum sum-rank distance) if $d(\Cc)=\srk(C)$ (respectively, $\maxsrk(\Cc)=\srk(C)$). Notice that a codeword with this property may not be unique.

Another useful code parameter is the covering radius. It has been defined for (non-linear) sum-rank metric codes with $m_1=\dots=m_{\ell}$ in~\cite[Definition 4]{CHW22}. Here, we extend the definition to any code in $\MM$.
\begin{definition}
The covering radius of a sum-rank metric code $\Cc\subseteq\MM$ is
\begin{equation*}
    \rho(\Cc)=\min\{r\in\mathbb{Z}:d(M,\Cc)\leq r\text{ for all }M\in\MM\},
\end{equation*}
where $d(M,\Cc)=\min\{d(M,C):C\in\Cc\}$.
\end{definition}
The following lemma collects some of the basic properties of these parameters.
\begin{lemma}
Let $\Cc\subseteq\mathcal{D}\subseteq\MM$ be sum-rank metric codes. Then
\begin{enumerate}
    \item $0\leq d(\Cc),\maxsrk(\Cc),\rho(\Cc)\leq n$.
    \item $d(\Cc)\geq d(\mathcal{D})$ and $\rho(\Cc)\geq \rho(\mathcal{D})$, while $\maxsrk(\Cc)\leq\maxsrk(\mathcal{D})$.
    \item $d(\Cc)\leq 2\rho(\Cc)+1$.
\end{enumerate}

\end{lemma}
\begin{example}\label{example:sumrankcode}
Consider the space $\MM=\F_2^{3\times 2}\times\F_2^{3\times 1}\times\F_2^{2\times 2}\times\F_2^{2\times 1}\times\F_2\times\F_2$ and let $\Cc\subseteq \MM$ be the following $\F_2$-linear sum-rank metric code
\begin{equation*}
    \Cc=\left\{\left(\begin{pmatrix}a_1&0\\0&a_1+a_2\\a_1&0\end{pmatrix},\begin{pmatrix}a_3\\0\\0\end{pmatrix},\begin{pmatrix}a_2&0\\0&a_2\end{pmatrix},\begin{pmatrix}a_4\\0\end{pmatrix},a_3,a_4 \right):(a_1,a_2,a_3,a_4)\in\F_2^4\right\}.
\end{equation*}
It is easy to verify that the minimum distance is equal to $2$ and it is realized by
\begin{equation*}
    C=\left(\begin{pmatrix}0&0\\0&0\\0&0\end{pmatrix},\begin{pmatrix}1\\0\\0\end{pmatrix},\begin{pmatrix}0&0\\0&0\end{pmatrix},\begin{pmatrix}0\\0\end{pmatrix},1,0 \right),
\end{equation*}
while the maximum sum-rank distance is equal to $7$ and it is realized by
\begin{equation*}
    D=\left(\begin{pmatrix}1&0\\0&0\\1&0\end{pmatrix},\begin{pmatrix}1\\0\\0\end{pmatrix},\begin{pmatrix}1&0\\0&1\end{pmatrix},\begin{pmatrix}1\\0\end{pmatrix},1,1 \right).
\end{equation*}
In order to compute the covering radius, one can consider the following element 
\begin{equation*}
    M=\left(\begin{pmatrix}0&1\\1&1\\0&0\end{pmatrix},\begin{pmatrix}0\\1\\0\end{pmatrix},\begin{pmatrix}1&1\\1&0\end{pmatrix},\begin{pmatrix}0\\1\end{pmatrix},1,1 \right)\in\MM.
\end{equation*}
One can check by direct computation that $d(M,\Cc)=6$ and therefore $\rho(\Cc)\geq6$. On the other hand we have that $\rho(\Cc)\leq6$, since for every element of $\MM$ we can find an element of $\Cc$ with the same last two entries. Therefore $\rho(\Cc)=6$.
\end{example}
We discuss now the notion of support in the sum-rank metric. Given an ambient space $\MM$, we consider the associated space  
$$ \mathbb{S}=\mathbb{F}_{q}^{n_1} \times \mathbb{F}_{q}^{n_2} \times \cdots \times \mathbb{F}_{q}^{n_\ell}.$$
We denote by $\mathcal{P}(\mathbb{S})$ the lattice of all subspaces $\mathcal{L}$ of $\mathbb{S}$ such that $\mathcal{L}=\mathcal{L}_1\times\dots\times\mathcal{L}_{\ell}$ where $\mathcal{L}_i$ is an $\F_q$-linear subspace of $\F_{q}^{n_i}$. 
\begin{definition}[{\cite[Definition 2.4]{BGRMSRD}}]
Let $\MM$, $\mathbb{S}$ and $\mathcal{P}(\mathbb{S})$ be as above. For an element $C\in\MM$, the support of $C$ is 
$$\supp(C)=\rowsp(C_1)\times\dots\times\rowsp(C_{\ell})\in\mathcal{P}(\mathbb{S}),$$
where $\rowsp(C_i)$ is the space generated by the rows of $C_i$ over $\F_q$. The support $\supp(\Cc)$ of a code $\Cc$ is the smallest $\mathcal{L}\in\mathcal{P}(\mathbb{S})$ such that $\supp(C)\subseteq\mathcal{L}$ for all $C\in\Cc$.
\end{definition}
\begin{remark}
If $\mathbb{S}=\F_q\times\dots\times \F_q$, then this definition of support coincides with the classical one for linear block codes. If $\mathbb{S}=\F_q^n$, then it coincides with the one for $\F_{q^m}$-linear rank-metric codes. In Section~\ref{sect:invars} we define equivalences (i.e., linear isometries) of sum-rank metric codes. We stress that this notion of support is not equivalence-invariant.
\end{remark}

\begin{definition}\label{defn:rowsupp}
For $\mathcal{L}\in\mathcal{P}(\mathbb{S})$  we define the row-support space $\mathcal{V}_{\mathcal{L}}$ as
$$\mathcal{V}_{\mathcal{L}} = \{ C \in \MM : \supp(C) \subseteq \mathcal{L} \}. $$
Clearly, $\mathcal{V}_{\mathbb{S}}=\MM$. 

For a code $\Cc\subseteq\MM$ and a subspace $\mathcal{L}\in\mathcal{P}(\mathbb{S})$ the subcode of $\Cc$ supported on $\mathcal{L}$ is
\begin{equation*}
    \Cc(\mathcal{L})=\Cc\cap\mathcal{V}_{\mathcal{L}}=\{C\in\Cc:\supp(C)\subseteq\mathcal{L}\}.
\end{equation*}
\end{definition}

\begin{example}
Let $\MM$, $\Cc$, $C$ and $D$ be as in Example~\ref{example:sumrankcode}. We have that $\supp(C)=0\times \F_2\times0\times0\times\F_2\times 0$, $\supp(D)=\langle (1,0)\rangle\times \F_2\times\F_2^2\times\F_2\times\F_2\times\F_2$ and $\supp(\Cc)=\MM$. Let $\mathcal{L}=\langle (1,0)\rangle\times 0\times \F_2^2\times 0\times 0\times 0$. Then,
\begin{equation*}
        \Cc(\mathcal{L})=\Cc\cap\mathcal{V}_{\mathcal{L}}=\left\langle\left(\begin{pmatrix}1&0\\0&0\\1&0\end{pmatrix},\begin{pmatrix}0\\0\\0\end{pmatrix},\begin{pmatrix}1&0\\0&1\end{pmatrix},\begin{pmatrix}0\\0\end{pmatrix},0,0 \right)\right\rangle_{\F_2}.
\end{equation*}
\end{example}
Let $\mathrm{tr}(M)$  denote the trace of a square matrix $M$ and consider the map $\mathrm{Tr}:\MM\times \MM\rightarrow\F_q$  given by $\mathrm{Tr}(D,C)=\sum_{i=1}^\ell \mathrm{tr}(D_i C_i^t) $. Notice that $\mathrm{Tr}$ is a non-degenerate symmetric bilinear form. The dual $\Cc^{\perp}$ of a sum-rank metric code $\Cc\subseteq\MM$ is the orthogonal subspace of $\Cc$ in $\MM$ with respect to $\mathrm{Tr}$.
\begin{definition}\label{defn:dual}
Let $\Cc \subseteq \matsumi$ be a code. The dual of $\Cc$ is 
$$\Cc^\perp = \{ D \in \matsumi : \mathrm{Tr}(D,C) = 0 \mbox{ for all } C \in \Cc\}.$$
\end{definition}
Some of the basic properties that we expect from the duality are trivially true since $\mathrm{Tr}$ is a non-degenerate symmetric bilinear form. For instance, $\lvert \MM\rvert=\lvert \Cc\rvert\lvert \Cc^{\perp}\rvert\text{ and }(\Cc^{\perp})^{\perp}=\Cc$.

\section{Isometries and invariants}\label{sect:invars}

We start this section by introducing the notion of equivalent codes in the sum-rank metric.
\begin{definition}[{\cite[Definition V.1]{ourpaper}}]
An isometry between two $\F_q$-linear sum-rank metric codes $\Cc_1$ and $\Cc_2$ is an $\F_q$-linear homomorphism $\varphi:\Cc_1\rightarrow\Cc_2$ that preserves the sum-rank metric, i.e., $\srk(\varphi(C)) = \srk(C)$ for all $C\in \Cc_1$.  In addition, two codes  $\Cc_1, \Cc_2 \subseteq \matsumi$ are equivalent if there exists an isometry of the ambient space
 $\varphi : \MM \rightarrow \MM$ such that $\varphi(\Cc_1) = \Cc_2$. 
\end{definition}
The next characterization of the isometries of $\MM$ was established in~\cite{ourpaper}. The proof follows easily from the classification of isometries in the rank-metric case. We recall that for every isometry $\psi:\F_q^{m\times n}\rightarrow \F_q^{m\times n}$ there exist two matrices $A$ and $B$ such that either $\psi(M)=AMB$ for all $M\in\F_q^{m\times n}$ or $\psi(M)=AM^tB$ for all $M\in\F_q^{m\times n}$, where the second case can only happen if $m=n$. 
\begin{theorem}[{\cite[Theorem V.2]{ourpaper}}]\label{thm:isom}
Let $\varphi : \matsumi \longrightarrow \matsumi$ be an $\F_q$-linear isometry. 
Then there is a permutation
$$\sigma : [\ell] \longrightarrow [\ell]$$
with the property that $\sigma(i) = j$ implies $m_i=m_j$ and $n_i = n_j$, and there are rank-metric $\F_q$-linear isometries $\psi_i : \F_q^{m_i \times n_i} \longrightarrow \F_q^{m_i \times n_i}$  for $i\in[\ell]$ such that 
$$\varphi(C_1 , \dots , C_\ell) = (\psi_1(C_{\sigma(1)}), \dots , \psi_\ell(C_{\sigma(\ell)}))$$ for all  $(C_1, \dots , C_\ell) \in \matsumi$.
\end{theorem}
It is well-known that every isometry between two linear block codes can be extended to an isometry of the ambient space.  This result is the renowned MacWilliams Extension Theorem \cite[Theorem 7.9.4]{huffman_pless_2003}. Since we know that a similar result does not hold in the case of rank-metric codes, it is evident that we cannot expect it to work in the sum-rank metric. However, in the latter case we have some new pathologies that arise from the structure of the ambient space, as one can appreciate in the next example.
\begin{example}
Let $\MM=\F_q^{2\times 2}\times  \F_q\times \F_q$.  Consider the following two codes
\begin{equation*}
\Cc_1=\left\{\left(\begin{pmatrix} a&0\\0&a\end{pmatrix},0,0\right):a\in\F_q\right\}\text{ and }\Cc_2=\left\{\left(\begin{pmatrix} 0&0\\0&0\end{pmatrix},a,a\right):a\in\F_q\right\}.
\end{equation*}
All the $\F_q$-linear isomorphisms between $\Cc_1$ and $\Cc_2$ are isometries and none of these can be extended to the ambient space.
\end{example}
More examples along these lines may be found in~\cite[Example V.4 and Example V.5]{ourpaper}.

\newpage
The notion of weight of a code is related to the size of its support.
\begin{definition}
The weight of a code $\Cc$ is
$$\wt(\Cc)=\min\{\maxsrk(\A):\Cc\subseteq\A=\A_1\times\ldots\times\A_\ell, \A_i\subseteq\F_q^{m_i\times n_i} \mbox{ is an optimal anticode for } i\in[\ell]\}.$$
A code $ \A \subseteq \mathbb{F}_q^{m \times n} $ is an optimal anticode if it attains the rank-metric anticode bound \cite[Proposition 47]{Rav16a}, i.e., $ \dim(\A) = m \maxrk(\A) $. See also Theorem \ref{theorem:anticodebound} below.
\end{definition}
Notice that if $m_1=\ldots=m_\ell=1$, i.e.,  if $\Cc$ is a linear block code, then the weight of $\Cc$ is the cardinality of its support. If $\ell=1$ and $n<m$, then the weight of $\Cc$ is the dimension of its support. 
Finally, if $\ell=1$ and $n=m$, then $\wt(\Cc)=\min\{\dim(\rowsp(\Cc)),\dim(\colsp(\Cc))\}$, where $\rowsp(\Cc)=\sum_{M\in\Cc}\rowsp(M)$ and $\colsp(\Cc)=\sum_{M\in\Cc}\colsp(M)$. This is related to the idea of defining the support of a square matrix $M$ as the pair $(\rowsp(M),\colsp(M))$, see also~\cite[Remark 11.2.4]{G21}.

Generalized weights are among the most studied invariants in coding theory. They were defined in 1977 by Helleseth, Kl\o ve and Mykkeltveit for linear block codes. The interest in these invariants exploded in 1991  when Wei showed that they capture the code performance in the wire-tap channel of type II. In~\cite{ourpaper} the authors gave the following definition for sum-rank metric codes, that is an extension of both the original definition and the one for rank metric codes given in~\cite{Rav16}.
\begin{definition}{\cite[Definition  VI.1]{ourpaper}}\label{definition:genwei}
Let $\Cc\subseteq \matsumi$ be a sum-rank metric code. For each $r\in[\dim(\Cc)]$, the  $r$-th generalized sum-rank weight of $\Cc$ is defined as
\begin{equation*}
d_r(\Cc) = \min\{\wt(\mathcal{D}):\mathcal{D}\text{ is a subcode of }\Cc\text{ and }\dim(\mathcal{D})\geq r\}.
\end{equation*}
We say that $\mathcal{D}\subseteq\Cc$ realizes $d_r(\Cc)$ if $\wt(\mathcal{D})=d_r(\Cc)$ and $\dim(\mathcal{D})\geq r$.
\end{definition}
If $m_1=\dots=m_\ell=m$, then the previous definition can be reformulated as
\begin{equation*}
\begin{split}
d_r(\Cc) =\frac{1}{m} \min\{\dim(\A):\,&\A= \A_1\times\dots\times \A_\ell \text{ where } \A_i\subseteq\F_q^{m\times n_i} \\  & \mbox{are optimal rank-metric anticodes and} \dim(\Cc\cap \A)\geq r\}.
\end{split}
\end{equation*}
Notice that the weight and generalized weights of a code are defined using direct products of optimal rank-metric anticodes. In the next section, we will see that direct products of optimal rank-metric anticodes are optimal sum-rank metric anticodes. Moreover, if $q\neq 2$ there are no other optimal sum-rank metric anticodes, while for $q=2$ there may be more. Therefore, limiting ourselves to these optimal anticodes is necessary in order to recover the classical definition for linear block codes in the case $q=2$. In the next example, we compute the generalized weights of a code and for each of them we exhibit a subcode that realizes it.
\begin{example}
Let $\MM=\F_q^{3\times 1}\times\F_q^{2\times 2}\times \F_q$. Consider the code $\Cc\subseteq \MM$ given by
\begin{equation*}
\Cc=\left\{\left(\begin{pmatrix} a_1\\a_2\\a_3\end{pmatrix},\begin{pmatrix} a_4&0\\0&a_4\end{pmatrix},a_3\right):(a_1,a_2,a_3,a_4)\in\F_q^4\right\}.
\end{equation*}
For $i\in[4]$, denote by $C_i$ the element of $\Cc$ that we obtain by setting $a_i=1$ and $a_j=0$ for $j\neq i$. One can verify that $d_1(\Cc)=d_2(\Cc)=1$, $d_3(\Cc)=2$, and $d_4(\Cc)=4$. The generalized weights are realized respectively by $\mathcal{D}_1=\langle C_1\rangle_{\F_q}$,  $\mathcal{D}_2=\langle C_1,C_2\rangle_{\F_q}$, $\mathcal{D}_3=\langle C_1,C_2,C_3\rangle_{\F_q}$, and $\mathcal{D}_4=\langle C_1,C_2,C_3,C_4\rangle_{\F_q}$.
\end{example}
Since an isometry of $\MM$ maps a product of optimal rank-metric anticodes into a product of optimal rank-metric anticodes with the same maximum sum-rank distance,  we have that generalized weights are invariant under equivalence. The next proposition collects some basic properties of generalized weights. In particular, it shows that they are non-decreasing and the first one always coincides with minimum distance of the code.
\begin{proposition}[\protect{\cite[Proposition VI.6]{ourpaper}}]\label{properties}
Let $0\neq\Cc \subseteq\mathcal{D} \subseteq \MM$, then:
\begin{enumerate}
\item $d_1(\Cc)= d(\Cc)$,
\item $d_r(\Cc)\leq d_s(\Cc)$ for $1\leq r\leq s\leq\dim(\Cc)$,
\item $d_r(\Cc)\geq d_r(\mathcal{D})$ for $r\in[\dim(\Cc)]$,
\item $d_{\dim(\Cc)}(\Cc)=\wt(\Cc)\leq n_1 + \dots + n_\ell$,
\end{enumerate}
\end{proposition}
Even though generalized weights are not strictly increasing, they still have to increase after a given number of steps. For instance, let $\Cc\subseteq\MM$ be a code and let $k\in[\ell]$, $r+m_k\in[\dim(\Cc)]$. In \cite[Lemma VI.8]{ourpaper} it is shown that if
\begin{equation}\label{eqn:growthgenwts}
d_{r+m_k}(\Cc)>\sum_{i=1}^{k-1}n_i\,\,\text{ then }\,\,d_{r+m_k}(\Cc)\geq d_r(\Cc)+1.
\end{equation}
In particular this implies that 
\begin{equation}\label{eqn:growthgenwts2}
d_{r+n_1m_1+\cdots+n_{j-1}m_{j-1}+\delta m_{j}}(\Cc)\geq d_r(\Cc)+n_1+\cdots+n_{j-1}+\delta
\end{equation}
for $j\in[\ell]$, $r\in[ \dim(\Cc)-(n_1m_1+\cdots+n_{j-1}m_{j-1}+\delta m_j)]$, and $0\leq\delta\leq n_j-1$. See \cite[Proposition VI.6]{ourpaper} for a proof of~\eqref{eqn:growthgenwts2}.

Notice that, together with Proposition~\ref{properties}, \eqref{eqn:growthgenwts2} allows us to explicitly compute the generalized weights of the ambient space $\MM$. In fact, since $\dim(\MM)=n_1m_1+\cdots+n_{\ell}m_{\ell}$,~\eqref{eqn:growthgenwts2} implies that
$$d_{n_1m_1+\cdots+n_{j-1}m_{j-1}+\delta m_{j}+s}(\MM)=n_1+\dots+n_{j-1}+\delta+1,$$
for $j\in[\ell]$, $0\leq\delta\leq n_j-1$, and $0<s\leq m_j$. 
In other words, the sequence of generalized weights of $\MM$ takes all the values between $1$ and $n_1+\ldots+n_{\ell}$ and the integer $n_1+\ldots+n_{j-1}+t$ with $t\in[n_j]$ appears exactly $m_{j}$ times.
Notice that, together with Proposition~\ref{properties}, \eqref{eqn:growthgenwts} implies that the sequence of generalized weights of a code $\Cc\subseteq\MM$ is a subsequence of that of $\MM$. 

Both in the Hamming metric and in the rank metric, the generalized weights of a code determine those of its dual. This result is broadly known as Wei's Duality Theorem, see~\cite[Theorem~3]{Wei} and~\cite[Corollary 38]{Rav16}. An extension of this theorem in the case of sum-rank metric codes with $m_1=\ldots=m_\ell=m$ is presented below. Let $\Cc \subseteq \F_q^{m\times n_1}\times\dots\times\F_q^{m\times n_{\ell}}$ be a sum-rank metric code. For each $r\in \Z$, let $D_r(\Cc)$ and $\overline{D}_r(\Cc)$ be the sets
$$D_r(\Cc) = 
\{ d_{r+sm}(\Cc) : s \in \Z , r+sm\in[\dim(\Cc)]\},$$
$$\overline{D}_r(\Cc) = \bigg\{ n+ 1 - d_{r+sm}(\Cc) : s \in \Z , r+sm\in[\dim(\Cc)]\bigg \}.$$
\begin{theorem}[\protect{\cite[Theorem VI.9]{ourpaper}}]\label{weiduality}
Let $m_1=\ldots=m_\ell=m$, $r\in[m]$, and let $\Cc \subseteq \MM$ be a sum-rank metric code. Then 
$$D_r(\Cd) = [n] \backslash \overline{D}_{r + \dim(\Cc)}(\Cc).$$
In particular, the generalized weights of a sum-rank metric code $\Cc$ determine the generalized weights of $\Cd$.
\end{theorem}

If we remove the assumption that $m_1=\ldots=m_\ell=m$, then it is no longer true in general that the generalized weights of a code determine those of its dual. This is shown in the next example.
\begin{example}\label{example:duality}
Consider the codes $\Cc_1,\Cc_2\subseteq\F_2^{2\times 2}\times \F_2$ given by
\begin{equation*}
        \Cc_1=\left\{\left(  \begin{pmatrix} 0 & 0 \\0 &0\end{pmatrix},a\right) : a \in \mathbb{F}_2\right\}\text{ and }\Cc_2=\left\{\left(  \begin{pmatrix} a & 0 \\0 &0\end{pmatrix},0\right) : a \in \mathbb{F}_2\right\}.
\end{equation*}
The corresponding duals are
\begin{equation*}
        \Cc_1^{\perp}=\left\{\left(  \begin{pmatrix} a & b \\c &d\end{pmatrix},0\right) : (a,b,c,d) \in \mathbb{F}_2^4\right\}\text{ and }\Cc_2^{\perp}=\left\{\left(  \begin{pmatrix} 0 & b \\c &d\end{pmatrix},a\right) : (a,b,c,d) \in \mathbb{F}_2^4\right\}.
 \end{equation*}
One has that $d_1(\Cc_1)=d_1(\Cc_2)=1$, while $d_4(\Cc_1^\perp) = 2$ and $d_4(\Cc_2^\perp) = 3$.
\end{example}

The following alternative definition of generalized weights was proposed and discussed in~\cite{ourpaper}.
\begin{definition}[{\cite[Appendix]{ourpaper}}]\label{definition:suppgenwei}
Let $\Cc\subseteq\mathbb{M}$ be a code.  The $r$-th generalized row-support weight of $\Cc$ is given by
\begin{equation*}
d^{Supp}_r(\Cc) = \min \left\{\dim(\mathcal{L}) :\,  \mathcal{L}\in\mathcal{P}(\mathbb{S})\text{ and } \dim \left( \Cc \cap \mathcal{V}_{\mathcal{L}}\right) \geq r\right\}
\end{equation*}
for $r\in[\dim(\mathcal{C})]$.
\end{definition}
Generalized row-support weights are an extension of the generalized weights introduced in~\cite{MM18} for rank-metric codes. They have a practical relevance in multishot linear network coding, since they measure information leakage to a wire-tapper. We refer the interested reader to~\cite[Appendix]{ourpaper} for more detail. Notice that generalized row-support weights are not in general invariant under equivalence. Indeed, if $n_i=m_i$ for some $i\in[\ell]$, a transposition in the $i$-th coordinate does not map a row-support space to a row-support space in general. However, the next proposition shows that in many cases this definition coincides with Definition~\ref{definition:genwei}.
\begin{proposition}[{\cite[Appendix]{ourpaper}}]
Let $\Cc\subseteq\MM$  be a code. If for all $i\in[\ell]$ we have $m_i> n_i$ or $ m_i=n_i= 1 $, then $ d^{Supp}_r(\Cc) = d_r(\Cc)$ for all $r\in[\dim(\Cc)]$.
\end{proposition}
Another important family of invariants of sum-rank metric codes is the sum-rank distribution.
\begin{definition}[\protect{\cite[Definition 2.8]{BGRMSRD}}]
Let $\Cc\subseteq\MM$ be a sum-rank metric code. For a non-negative integer $r$, let
\begin{equation*}
W_r(\Cc)=\left\lvert\{C\in\Cc:\srk(\Cc)=r\}\right\rvert.
\end{equation*} 
The sequence $\left(W_r(\Cc)\right)_{r>0}$ is the sum-rank distribution of $\Cc$. 
\end{definition}
The notion of sum-rank distribution generalizes that of weight distribution for linear block codes and rank-metric codes. In contrast to what happens in those two cases, we do not have MacWilliams Identities for sum-rank metric codes, as was first observed in \cite[Example 5.2]{BGRMSRD}.
The following is a simple example of this phenomenon.
\begin{example}
Let $\Cc_1,\Cc_2\subseteq\F_2^{2\times 2}\times \F_2$ be as in Example~\ref{example:duality}. Then, $\Cc_1$ and $\Cc_2$ have the same sum-rank distribution. Indeed, $W_0(\Cc_1)=W_0(\Cc_2)=W_1(\Cc_1)=W_1(\Cc_2)=1$ and $W_1(\Cc_i)=W_1(\Cc_i)=0$ for $i> 1$.  Instead for $\Cc_1^\perp,\Cc_2^\perp$ we have $W_1(\Cc_1^\perp)=9$ and $W_1(\Cc_2^\perp)=6$.
\end{example}

It is nevertheless possible to define other partitions of a sum-rank metric code for which there exists a relation between the distribution of a code and the one of its dual.
\begin{definition}[\protect{\cite[Definition 2.8]{BGRMSRD}}]
 For a vector $v\in\mathbb{Z}_{\geq0}^{\ell}$, let
\begin{equation*}
W_v(\Cc)=\left\lvert\{C\in\Cc:\rk(C_i)=v_i\text{ for all }i\in[\ell]\}\right\rvert.
\end{equation*} 
The list $\left(W_v(\Cc)\right)_{v\in\mathbb{Z}_{\geq0}^{\ell}}$ is called the rank-list distribution of $\Cc$.
\end{definition}
In the next example we compute the rank-list distribution of the codes of Example~\ref{example:duality}.
\begin{example}
Let $\Cc_1,\Cc_2\subseteq\F_2^{2\times 2}\times \F_2$ be as in Example~\ref{example:duality}. The rank-list distribution of $\Cc_1$ is given by
$W_{(0,0)}(\Cc_1)=W_{(0,1)}(\Cc_1)=1$ and $W_u(\Cc_1)=0$ when $u\neq(0,0),(0,1)$. The rank-list distribution of $\Cc_1^{\perp}$ is $W_{(0,0)}(\Cc_1^{\perp})=1$, $W_{(1,0)}(\Cc_1^{\perp})=9$, $W_{(2,0)}(\Cc_1^{\perp})=6$ and $0$ in all the remaining cases. The rank-list distribution of $\Cc_2$ is $W_{(0,0)}(\Cc_1)=W_{(1,0)}(\Cc_1)=1$ and $W_u(\Cc_1)=0$ when $u\neq(0,0),(1,0)$. The rank-list distribution of $\Cc_2^{\perp}$ is $W_{(0,0)}(\Cc_2^{\perp})=1$, $W_{(1,0)}(\Cc_2^{\perp})=5$, $W_{(0,1)}(\Cc_2^{\perp})=1$, $W_{(2,0)}(\Cc_2^{\perp})=2$, $W_{(1,1)}(\Cc_2^{\perp})=5$, $W_{(2,1)}(\Cc_2^{\perp})=2$  and $0$ in the remaining cases.
\end{example}
The next theorem is the analogue of the MacWilliams Identities for the rank-list distribution.
Recall that the $q$-ary Gaussian coefficient of $a,b\in\mathbb{Z}$ is defined as
\begin{equation*}
\gauss{a}{b}_q=\begin{cases} 0&\text{if }a<0,\,b<0,\text{ or }b>a\\
1&\text{if }b=0\text{ and }a\geq 0\\
\frac {(q^a-1)(q^{a-1}-1)\cdots(q^{a-b+1}-1)}{(q^b-1)(q^{b-1}-1)\cdots(q-1)}&\text{otherwise}.\end{cases}
\end{equation*}
\begin{theorem}[\protect{\cite[Theorem 5.5]{BGRMSRD}}]
Let $\Cc\subseteq\MM$ be a code. Then
\begin{equation*}
W_v(\Cc^{\perp})= \frac {1}{\lvert \Cc\rvert}\sum_{u\in\mathbb{Z}_{\geq0}^{\ell}}W_u(\Cc)\sum_{w\leq v}q^{\sum_{i=1}^{\ell}m_iw_i}\prod_{i=1}^{\ell}(-1)^{v_i-w_i}q^{\binom{v_i-w_i}{2}}\gauss{n_i-u_i}{w_i}_q\gauss{n_i-w_i}{v_i-w_i}_q
\end{equation*}
for $v\in\mathbb{Z}_{\geq0}^{\ell}$.
\end{theorem}
The next theorem is the sum-rank metric analogue of the binomial moments of MacWilliams identities.
\begin{theorem}[\protect{\cite[Theorem 5.6]{BGRMSRD}}]
Let $\Cc\subseteq\MM$ be a code. Then
\begin{equation*}
\sum_{u\in\mathbb{Z}_{\geq0}^{\ell}}W_u(\Cc)\prod_{i=1}^{\ell}\gauss{n_i-u_i}{v_i-u_i}_q=\frac{\lvert\Cc\rvert}{q^{\sum_{i=1}^{\ell}m_i(n_i-v_i)}}\sum_{u\in\mathbb{Z}_{\geq0}^{\ell}}W_u(\Cc^{\perp})\prod_{i=1}^{\ell}\gauss{n_i-u_i}{v_i}_q
\end{equation*}
for $v\in\mathbb{Z}_{\geq0}^{\ell}$.
\end{theorem}
Notice that for $\ell=1$ the previous equality becomes
\begin{equation*}
\sum_{i=0}^{r}W_i(\Cc)\gauss{n-i}{r-i}_q=\frac{\lvert\Cc\rvert}{q^{m(n-r)}}\sum_{i=0}^{n-r}W_i(\Cc^{\perp})\gauss{n-i}{r}_q,
\end{equation*} 
for $r\in[n]$. This is exactly the identity proved in~\cite[Theorem 31]{Rav16a} for rank metric codes. Moreover in the case $m_1=\dots=m_{\ell}=1$ we obtain 
\begin{equation*}
 \sum_{\substack{u\in\{0,1\}^{\ell}\\ u\leq v}}W_u(\Cc)=\frac{\lvert \Cc\rvert}{q^{\ell-\wt(v)}}\sum_{\substack{u\in\{0,1\}^{\ell}\\u\leq(1,\dots,1)-v}}W_{u}(\Cc^{\perp})
\end{equation*}
for $v\in\{0,1\}^{\ell}$.
Summing over all $v$ with $\wt(v)=g$, one finds 
\begin{equation*}
\sum_{i=0}^{g}W_i(\Cc)\binom{\ell-i}{\ell-g}=\frac{\lvert \Cc\rvert}{q^{\ell-g}}\sum_{i=0}^{\ell-g}W_i(\Cc^{\perp})\binom{\ell-i}{g},
\end{equation*}
which is the equation given in~\cite[equation (M1), page 257]{huffman_pless_2003} for linear block codes. 

We conclude this section by discussing an additional definition of distribution for sum-rank metric codes.
\begin{definition}[\protect{\cite[Definition 2.8]{BGRMSRD}}]
Let $\Cc\subseteq\MM$ be a code. For $\mathcal{L}\in\mathcal{P}(\mathbb{S})$, let 
\begin{equation*}
W_{\mathcal{L}}(\Cc)=\left\lvert\{C\in\Cc:\supp(C)=\mathcal{L}\}\right\rvert.
\end{equation*} 
The list $\left(W_{\mathcal{L}}(\Cc)\right)_{\mathcal{L}\in\mathcal{P}(\mathbb{S})}$ is the support distribution of $\Cc$.
\end{definition}
Similarly to the rank-list distribution, the support distribution of a code determines that of its dual.
\begin{theorem}[\protect{\cite[Theorem 5.4]{BGRMSRD}}]
Let $\Cc\subseteq\MM$ be a code. Let $\mathcal{L}=\mathcal{L}_1\times\dots\times\mathcal{L}_{\ell}\in\mathcal{P}(\mathbb{S})$ and $v=(\dim(\mathcal{L}_1),\dots,\dim(\mathcal{L}_{\ell}))$. Then
\begin{equation*}
W_{\mathcal{L}}(\Cc^{\perp})=\frac{1}{\lvert\Cc\rvert}\sum_{\mathcal{H}\in\mathcal{P}(\MM)}W_{\mathcal{H}}(\Cc)\sum_{u\leq v}q^{\sum_{i=1}^{\ell}m_iu_i}\prod_{i=1}^{\ell}(-1)^{v_i-u_i}q^{\binom{v_i-u_i}{2}}\gauss{\dim(\mathcal{H}_i\cap\mathcal{L}_i)}{u_i}_q.
\end{equation*}
\end{theorem}

\section{Bounds, optimal anticodes, and MSRD codes}\label{sect:extremal}
The concept of $r$-anticode was first introduced in~\cite[Definition 2.1]{BGROAC} in the wider context of non-linear sum-rank metric codes. Here, we recall the definition in the linear case.  
\begin{definition}
For $r\in[n]$, a sum-rank metric code $\Cc\subseteq\MM$ is an $r$-anticode if $\maxsrk(\Cc)\leq r$.
\end{definition}
A first upper bound on the dimension of an $r$-anticode appears in~\cite[Theorem 2.2]{BGROAC}. For $r\in[n]$, let $0\leq i<\ell$ and $0\leq\delta<n_{i+1}$ be such that $r=n_1+\dots+n_i+\delta$. Notice that $i$ and $\delta$ are uniquely determined by $r$. Then, if $\Cc$ is an $r$-anticode,
\begin{equation}\label{equation:oldanticodebound}
    \dim(\Cc)\leq m_1n_1+\dots+m_in_i+m_{i+1}\delta.
\end{equation}
A stronger bound was independently proved in~\cite[Theorem IV.1]{ourpaper}. Its proof is related to the theorem by Meshulam~\cite[Theorem~2]{Mes85} and its generalization to cosets of vector spaces by de Seguins Pazzis~\cite[Corollary 2]{Pazzis2010TheAP}.
In order to see that the next bound is tighter than~\eqref{equation:oldanticodebound}, it suffices to observe that $\sum_{i=1}^\ell m_i \rk(C_i)\leq m_1n_1+\dots+m_in_i+m_{i+1}\delta$ for every $C\in\MM$ with $\srk(C)=r$.
\begin{theorem}[Anticode Bound, {\cite[Theorem IV.1]{ourpaper}}]\label{theorem:anticodebound}
Let $\Cc\subseteq \MM$ be a code. Then 
\begin{equation*}
\dim(\Cc)\leq \max_{C\in\Cc} \left\{\sum_{i=1}^\ell m_i \rk(C_i)\right\}.
\end{equation*}
\end{theorem}
If $m_1=\ldots=m_\ell=m$, then the Anticode Bound can be expressed more concisely as $$\dim(\Cc)\leq m\maxsrk(\Cc).$$
Starting from this formula, it is straightforward to verify that Theorem~\ref{theorem:anticodebound} is an extension of both the rank-metric and the Hamming-metric Anticode Bound. Indeed, if $\ell=1$ we immediately obtain $\dim(\Cc)\leq m\mathrm{maxrk}(\Cc)$, while if $m=1$ we have $\dim(\Cc)\leq\mathrm{maxwt}(\Cc)$, where the $\mathrm{maxwt}(\Cc)$ is the maximum Hamming weight of a codeword in $\Cc$.
\begin{definition}[{\cite[Definition IV.2]{ourpaper}}]\label{definition:optimalanticode}
A sum-rank metric code $\Cc\subseteq\MM$ is an optimal anticode if attains the bound in Theorem~\ref{theorem:anticodebound}, i.e.,
\begin{equation*}
\dim(\Cc)=\max_{C\in\Cc}\left\{ \sum_{i=1}^\ell m_i \rk(C_i)\right\}.
\end{equation*}
\end{definition} 
\begin{remark}
Optimal anticodes are defined in~\cite{BGROAC} as the $r$-anticodes that attain bound~\eqref{equation:oldanticodebound}. A complete classification of these codes is given in~\cite[Corollary 3.8]{BGROAC}. Since the bound in Theorem~\ref{theorem:anticodebound} is tighter than~\eqref{equation:oldanticodebound}, it is clear that the codes in this family are also optimal anticodes according to Definition~\ref{definition:optimalanticode}. The converse is not true, as the next example shows.
\end{remark}
\begin{example}
 Let $\MM=\F_2^{3\times 2}\times\F_2^{3\times 1}\times\F_2^{2\times 2}\times\F_2^{2\times 1}\times\F_2\times\F_2$. The code $\Cc=0\times0\times\F_2^{2\times2}\times 0\times 0\times 0$ is an optimal anticode of dimension $4$ and $\maxsrk(\Cc)=2$. However, $\Cc$ does not meet the bound in~\eqref{equation:oldanticodebound} since $\dim(\Cc)=4<6=m_1n_1$. Therefore, $\Cc$ is an optimal anticode according to Definition~\ref{definition:optimalanticode}, but not according to~\cite[Definition 2.3]{BGROAC}.
\end{example}

The next theorem provides a complete classification of optimal anticodes in the sum-rank metric as products of optimal anticodes in the rank and the Hamming metric.

\begin{theorem}[Classification of optimal anticodes, {\cite[Theorem IV.11]{ourpaper}}]\label{thm:classificationanticodes}
Let $k=0$ if $m_1=1$ and $k=\max\{i\in[\ell]\mid m_i>1\}$ otherwise. 
A code $\Cc\subseteq\MM$ is an optimal anticode if and only if there is an optimal anticode $\Cc^\prime\subseteq\F_q^{\ell-k}$ and optimal anticodes $\Cc_i\subseteq\F_q^{m_i\times n_i}$ for all $i\in[k]$ such that $\Cc=\prod_{i=1}^k\Cc_i\times\Cc^\prime$. 
\end{theorem}

When $q\neq2$ or $n\leq 2$, an optimal Hamming-metric anticode in $\F_q^n$ is a direct product of copies of $\F_q$ and $0$. Therefore, if we restrict to that situation, we can reformulate the previous theorem as follows.

\begin{corollary}[{\cite[Corollary IV.13]{ourpaper}}]
Assume that either $q\neq2$ or $m_{\ell-2} \geq 2$. An $\F_q$-linear space $\Cc\subseteq\MM$ is an optimal anticode if and only if for all $i\in[\ell]$ there is $\Cc_i\subseteq\mathbb{F}_q^{m_i\times n_i}$ optimal anticode such that $\Cc=\prod_{i=1}^\ell\Cc_i$. 
\end{corollary}
In particular, one obtains the following results.
\begin{proposition}[{\cite[Proposition IV.14]{ourpaper}}]\label{prop:dualOAC}
Let $q\neq 2$ or $ m_{\ell-2}\geq 2$. Then $\A \subseteq \matsumi$  is an optimal anticode if and only if $\A^\perp \subseteq \matsumi$ is an optimal anticode.
\end{proposition}

It is easy to show that, if $q\neq 2$ or $m_{\ell-2} \geq 2$, an optimal sum-rank metric anticode is generated by its elements of sum-rank weight one. Moreover, one has the following.

\begin{theorem}[{\cite[Theorem IV.7]{ourpaper}}]
Assume that either $q\neq 2$ or $m_{\ell-1} \geq 2$ and let $\Cc\subseteq\MM$ be an optimal anticode. Then $\Cc$ is generated by its elements of maximum sum-rank weight. 
\end{theorem}

Theorem \ref{thm:classificationanticodes} allows us to write an explicit formula for the generalized weights of most optimal anticodes.

\begin{theorem}\label{thm:genwtsOAC}
Let $\Cc=\prod_{i=1}^\ell\Cc_i\subseteq\MM$ be an optimal sum-rank metric anticode, where $\Cc_i\subseteq\F_q^{m_i\times n_i}$ is an optimal anticode for $i\in[\ell]$ with $t_i=\maxrk(\Cc_i)$. Let $r\in[\dim(\Cc)]$ and write $r=\sum_{i=1}^{j-1}m_it_i+m_j\tau+\sigma+1$, where $j\in[\ell]$, $0\leq\tau\leq t_j-1$, and $0\leq\sigma\leq m_j-1$.
Then
$$d_r(\Cc)=\sum_{i=1}^{j-1}t_i+\tau+1.$$
\end{theorem}

It easily follows from the definition that a rank-metric code or a non-binary linear block code which has the generalized weights of an optimal anticode is an optimal anticode. However, this is not the case in general for sum-rank metric codes, as the next example shows.

\begin{example}
Let $\Cc_1,\Cc_2\subseteq\F_2^{2\times 2}\times \F_2$ be the codes of Example \ref{example:duality}. 
One has that $d_1(\Cc_1)=d_1(\Cc_2)=1$. However, $\Cc_1$ is an optimal anticode, while $\Cc_2$ is not.
\end{example}

The Singleton Bound involves the size of the ambient space and the dimension and minimum distance of a code. For (not necessarily linear) sum-rank metric codes, a Singleton Bound relating the size of the ambient space, the minimum distance, and the cardinality of the code was shown in~\cite[Theorem 3.2]{BGRMSRD}. 
For linear sum-rank metric codes, the next theorem is more general, as it provides a bound on the dimension of a code, based on the value of its $r$-th generalized weight.
\begin{theorem}[Singleton Bound, {\cite[Theorem VII.4]{ourpaper}}]\label{singletonbound}
Let $\Cc\subseteq\MM$ be a code and let $r\in[\dim(\Cc)]$. Let $j\in[\ell]$ and $0\leq\delta\leq n_j-1$ be such that $$d_r(\Cc)-1\geq\sum_{i=1}^{j-1} n_i+\delta.$$
Then
$$
\dim(\Cc)\leq\sum_{i=j}^\ell m_in_i-m_{j}\delta+r-1.$$
\end{theorem}

Theorem \ref{singletonbound} yields upper bounds on all the generalized weights of $\Cc$.

\begin{corollary}[{\cite[Corollary VII.5]{ourpaper}}]\label{cor.rsb}
Let $\Cc\subseteq\MM$ be a code. Let $r\in[\dim(\Cc)]$, $j\in[\ell]$, and $0\leq\delta\leq n_j-1$ be such that $\dim(\Cc)\geq\sum_{i=j}^\ell m_in_i-m_{j}\delta+r$. Then
$$d_r(\Cc)\leq\sum_{i=1}^{j-1} n_i+\delta.$$
In particular, if $\dim(\Cc)=\sum_{i=j}^\ell m_in_i-m_{j}\delta$, then $$d_1(\Cc)\leq\ldots\leq d_{m_j}(\Cc)\leq \sum_{i=1}^{j-1} n_i+\delta+1.$$
\end{corollary}

In view of Corollary \ref{cor.rsb}, it is natural to define MSRD codes as follows. 

\begin{definition}[{\cite[Definition 3.3]{BGRMSRD}} and {\cite[Definition VII.6]{ourpaper}}]
A code $\Cc$ is MSRD if there exist $j\in[\ell]$ and $0\leq \delta\leq n_j-1$ such that $$d(\Cc)=\sum_{i=1}^{j-1} n_i+\delta+1 \;\mbox{ and }\; \dim(\Cc)=\sum_{i=j}^{\ell} m_in_i-\delta m_{j}.$$
\end{definition}

The next proposition collects some properties which are equivalent to being MSRD.

\begin{proposition}[{\cite[Proposition VII.7 and Proposition VII.8]{ourpaper}}]
Let $\Cc\subseteq\MM$ be a code. The following are equivalent:
\begin{itemize}
\item $\Cc$ is MSRD.
\item Write $d(\Cc)=\sum_{i=1}^{j-1}n_i+\delta+1$, where $j\in[\ell]$ and $0\leq\delta\leq n_j-1$. For any $\mathcal{A}$ optimal sum-rank anticode of $\maxsrk(\A)=d(\Cc)-1$ and $\dim(\A)=\sum_{i=1}^{j-1} m_in_i+(\delta+1)m_j$, one has $\Cc+\A=\MM$.
\item There are $j\in[\ell]$ and $0\leq\delta\leq n_j-1$ such that $\dim(\Cc)=\sum_{i=j}^\ell m_in_i-m_j\delta$ and for any $\mathcal{A}$ optimal anticode of $\maxsrk(\A)\leq \sum_{i=1}^{j-1}n_i+\delta$ one has $\Cc\cap\A=0$. 
\item For $S\subseteq[n]$, denote by $\F_q[S]$ the set of elements of $\MM$ which are zero outside of the columns indexed by $S$.
For any $d(\Cc)\leq h\leq n$, let $S_h:=[d(\Cc)-1]\cup\{h\}$. Then 
$$\dim(\Cc\cap\mathbb{F}_q[S_h])=m_{k},$$ 
where $k=\max\{\nu\ |\ \sum_{i=1}^{\nu-1} n_i<h\}$.
\end{itemize}
\end{proposition}

We refer the interested reader to \cite[Section VII]{ourpaper} for a detailed discussion of conditions which are equivalent to or stronger than the property of being MSRD.

It is well known that the dual of an MDS linear block code is MDS and the dual of an MRD rank-metric code is MRD.
More generally, if $m_1=\ldots=m_\ell$, then the dual of an MSRD sum-rank metric code is MSRD. However, in the generality of sum-rank metric codes, this is essentially the only situation in which the dual of an MSRD code is MSRD. 

\begin{theorem}[{\cite[Proposition VII.12 and Proposition VII.13]{ourpaper}}]
If $m_1=\ldots=m_\ell$, then the dual of an MSRD code is MSRD. Conversely, if there exists a non-trivial code $\Cc\subseteq\MM$ such that both $\Cc$ and $\Cc^\perp$ are MSRD, then $m_1=\cdots= m_\ell$.
\end{theorem}

Finally, the property that both $\Cc$ and $\Cc^\perp$ are MSRD has a simple equivalent numerical formulation.

\begin{proposition}[{\cite[Proposition VII.11]{ourpaper}}]
Let $\Cc\subseteq\MM$ be a non-trivial code. The following are equivalent:
\begin{itemize}
\item $\Cc$ and $\Cc^\perp$ are MSRD,
\item $d(\Cc)+d(\Cc^\perp)=n+2$.
\end{itemize}
\end{proposition}

Both for linear block codes and rank-metric codes, the generalized weights of an MDS or MRD code are determined by its parameters. 
The next theorem extends this result to MSRD codes. In particular, the generalized weights of an MSRD code $\Cc$ are determined by its parameters and are the largest possible for a subcode of $\MM$ of the same dimension as $\Cc$. In fact, they coincide with the last $\dim(\Cc)$ generalized weights of $\MM$.

\begin{theorem}[\protect{\cite[Theorem VII.14]{ourpaper}}]\label{genwtsMSRD}
Let $\Cc\subseteq\MM$ be an MSRD code and write $d(\Cc)=\sum_{i=1}^{j-1} n_i+\delta+1$ for some $j\in[\ell]$ and $0\leq \delta\leq n_j-1$. Then for $r\in[\dim(\Cc)]$
$$d_r(\Cc)=d_{\sum_{i=1}^{j-1}m_in_i+m_j\delta+r}(\MM).$$
\end{theorem} 

Alternatively, one can express the generalized weights of an MSRD code as follows.
Let $h\in[n]$ and write $h=\sum_{i=1}^{t-1}n_i+\delta+1$, with $0\leq \delta\leq n_t-1$.
For each $h\in[n]$ we define $r_h$ as $$r_h=\sum_{i=1}^{t-1} m_in_i+(\delta+1)m_t.$$ 
For an MSRD code $\Cc$, let $d(\Cc)\leq h\leq n$ and let $k=\max\{\nu\mid\sum_{i=1}^{\nu-1}n_i<h\}$. Theorem~\ref{genwtsMSRD} states that 
$$d_{r_h-r_{d(\Cc)-1}-m_k+1}(\Cc)=\cdots=d_{r_h-r_{d(\Cc)-1}}(\Cc)=h.$$

A similar result holds for $r$-MSRD codes. $r$-MSRD codes are a generalization of $r$-MRD codes in the rank metric and are defined as follows. 
\begin{definition}[\protect{\cite[Definition VII.18]{ourpaper}}]\label{defn:rMSRD}
Let $j\in[\ell]$, $0\leq \delta\leq n_j-1$, and let $\Cc\subseteq\MM$ be a code of $\dim(\Cc)=\sum_{i=j}^\ell m_in_i-\delta m_j$. 
Define $d_{\max}=\sum_{i=1}^{j-1}n_i+\delta+1$. Let $d_{\max}\leq h\leq n$ and $$r=r_h-r_{d_{\max}-1}-m_k+1,$$ 
where $k=\max\{\nu\mid \sum_{i=1}^{\nu-1}n_i<h\}$.
We say that $\Cc$ is $r$-MSRD if $d_r(\Cc)=h.$
\end{definition}
Notice that a code is MSRD if and only if it is $1$-MSRD. Therefore, the next theorem extends Theorem~\ref{genwtsMSRD}.
\begin{theorem}[\protect{\cite[Theorem VII.19]{ourpaper}}]\label{thm:rMSRD}
Let $j\in[\ell]$, $0\leq\delta\leq n_j-1$, and let $\Cc\subseteq\MM$ be a non-trivial code of $\dim(\Cc)=\sum_{i=j}^\ell m_in_i-\delta m_j$. Define $d_{\max}=\sum_{i=1}^{j-1}n_i+\delta+1$. Let $d_{\max}\leq h\leq n$ and $$r=r_h-r_{d_{\max}-1}-m_k+1,$$ 
where $k=\max\{\nu\mid \sum_{i=1}^{\nu-1}n_i<h\}$. If $\Cc$ is $r$-MSRD, then  $$d_{r}(\Cc)=\ldots=d_{r+m_k-1}(\Cc)=h.$$
Moreover, $\Cc$ is $(r+m_k)$-MSRD. 
\end{theorem}
Similarly to MSRD codes, any $r$-MSRD code $\Cc$ has the property that its generalized weights $d_r(\Cc),\ldots,d_{\dim(\Cc)}(\Cc)$ are determined by its parameters and are as large as possible, for a subcode of $\MM$ of the same dimension as $\Cc$. In fact, they coincide with the corresponding generalized weights of an MSRD code of the same dimension as $\Cc$.

We conclude the section by discussing the existence of MSRD codes. In the next section, we discuss explicit MSRD code constructions that are linear over an extension field, including linearized Reed-Solomon codes \cite{linearizedRS}. Such constructions satisfy the condition $ m_1 = m_2 = \ldots = m_\ell $. Non-trivial constructions of MSRD codes such that not all $ m_1, m_2 , \ldots, m_\ell $ are equal are given in \cite[Section VII]{BGRMSRD}. In addition, the next results allow us to construct MSRD codes by puncturing and shortening MSRD codes.

\begin{theorem}[\protect{\cite[Theorem VI.14]{BGRMSRD}}] \label{th shortening MSRD}
Suppose there exists an MSRD code $\Cc\subseteq\MM$ such that $$d(\Cc)=\sum_{i=1}^{j-1} n_i+\delta+1 \;\mbox{ and }\; \dim(\Cc)=\sum_{i=j}^{\ell} m_in_i-\delta m_{j},$$ for some $j\in[\ell]$ and $0\leq \delta\leq n_j-1$.
\begin{enumerate}
\item
Choose $ s \in \{j,\ldots, \ell \} $ and set
$$ \widetilde{n}_i = \left\lbrace \begin{array}{ll}
n_i, & \textrm{if } i \neq s, \\
n_s - 1, & \textrm{if } i = s.
\end{array} \right. $$
There exists an MSRD code with sum-rank distance $ d $ in $ \widetilde{\MM} := \mathbb{F}_q^{m_1 \times \widetilde{n}_1} \times \cdots \times \mathbb{F}_q^{m_\ell \times \widetilde{n}_\ell} $. 
\item
Choose $ s \in \{j + 1,\ldots, \ell \} $ and set
$$ \widetilde{m}_i = \left\lbrace \begin{array}{ll}
m_i, & \textrm{if } i \neq s, \\
m_s - 1, & \textrm{if } i = s.
\end{array} \right. $$
There exists a linear MSRD code with sum-rank distance $ d $ in $ \widetilde{\MM} := \mathbb{F}_q^{\widetilde{m}_1 \times n_1} \times \cdots \times \mathbb{F}_q^{\widetilde{m}_\ell \times n_\ell} $. 
\end{enumerate}
\end{theorem}

\begin{theorem}[\protect{\cite[Theorem VI.15]{BGRMSRD}}] \label{th puncturing MSRD}
Suppose there exists an MSRD code $\Cc\subseteq\MM$ such that $$d(\Cc)=\sum_{i=1}^{j-1} n_i+\delta+1 \;\mbox{ and }\; \dim(\Cc)=\sum_{i=j}^{\ell} m_in_i-\delta m_{j},$$ for some $j\in[\ell]$ and $0\leq \delta\leq n_j-1$. If
$ \delta > 0 $, choose $ s \in [j] $ and if $ \delta = 0 $, choose $ s \in [j-1] $. Set
$$ \widetilde{n}_i = \left\lbrace \begin{array}{ll}
n_i, & \textrm{if } i \neq s, \\
n_s - 1, & \textrm{if } i = s.
\end{array} \right. $$
There exists an MSRD code with sum-rank distance $ d-1 $ in $ \widetilde{\MM} := \mathbb{F}_q^{m_1 \times \widetilde{n}_1} \times \cdots \times \mathbb{F}_q^{m_\ell \times \widetilde{n}_\ell} $. 
\end{theorem}

Theorem \ref{th shortening MSRD} is obtained via shortening, whereas Theorem \ref{th puncturing MSRD} is obtained via puncturing (or restriction). These operations may be performed e.g. on the MSRD constructed in \cite{linearizedRS}, \cite[Section VII]{BGRMSRD}, and in the next section to obtain MSRD codes with new parameters. Notice that when $ m_1 = \ldots = m_\ell $, the shortening and puncturing above can be done on any index $s$, which recovers \cite[Corollary 7]{gsrws}. Notice however that the construction in \cite[Corollary 7]{gsrws} preserves $ \mathbb{F}_{q^m} $-linearity, as discussed in the next section.

The next is a non-existential results, in the form of a bound on the parameter $ \ell $ for MSRD codes.

\begin{theorem} [\protect{\cite[Theorem VI.12]{BGRMSRD}}] \label{th non existence MSRD}
Suppose $ \nu = n_1 = \ldots = n_\ell $ and $ m = m_1 = \ldots = m_\ell $, and suppose there exists an MSRD code $\Cc\subseteq\MM$ of minimum sum-rank distance $ d \geq 3 $. Then
\begin{equation*}
\begin{split}
\ell & \leq \left\lfloor \frac{d-3}{\nu} \right\rfloor + \left\lfloor \frac{q^\nu - q^{\nu \lfloor (d-3)/\nu \rfloor+\nu-d+3} + (q-1)(q^m+1)}{q^\nu-1} \right\rfloor \\
 & \leq \left\lfloor \frac{d-3}{\nu} \right\rfloor + 1 + \left\lfloor \frac{q^m(q-1)}{q^\nu-1} \right\rfloor .
\end{split}
\end{equation*}
In particular, we have the following:
\begin{enumerate}
\item
If $ \nu \mid d-3 $, then
$$ \ell \leq \frac{d-3}{\nu} + \left\lfloor \frac{(q-1)(q^m+1)}{q^\nu-1} \right\rfloor . $$ 
\item
If $ d \leq \nu+2 $, then
\begin{equation*}
\begin{split}
\ell & \leq \left\lfloor \frac{q^\nu - q^{\nu-d+3} + (q-1)(q^m+1)}{q^\nu-1} \right\rfloor \\
 & \leq 1 + \left\lfloor \frac{q^m(q-1)}{q^\nu-1} \right\rfloor .
\end{split}
\end{equation*}
If in addition $ \nu = m $, then 
$$ \ell \leq \left\lfloor \frac{q^{\nu+1}-1}{q^\nu-1} \right\rfloor \leq q+1, $$
and if $ \nu = m \leq 2 $, then $ \ell \leq q $.
\item
If $ d = 3 $ and $ \nu \mid m $, then
$$ \ell \leq (q-1) \cdot \frac{q^m-1}{q^\nu-1}. $$
\end{enumerate}
\end{theorem}

Assume that $ \nu = n_1 = \ldots = n_\ell $ and $ m = m_1 = \ldots = m_\ell $. In this case, the MDS Conjecture \cite[page 265]{huffman_pless_2003} implies that, if an MSRD code exists, then $ \ell \leq (q^m+1)/n $. However, Theorem \ref{th non existence MSRD} gives a tighter bound on $ \ell $ when $ d \leq q^m(1-n(q-1)/(q^n-1))+4-n $. If $ m=\nu $ and $ d \leq \nu+2 $, then Theorem \ref{th non existence MSRD} gives the tighter bound $ \ell \leq q+1 $. Notice that, for any $ d \geq 3 $, if $ m=\nu=1 $, then the sum-rank metric coincides with the Hamming metric and Theorem \ref{th non existence MSRD} yields the bound $ \ell \leq q+d-2 $, which is known for MDS codes, see \cite[Corollary 7.4.3(ii)]{huffman_pless_2003}.

Finally, the bound in item 3 in Theorem \ref{th non existence MSRD} is met with equality by the MSRD codes in \cite[Section 4.4]{generalMSRD}. Moreover, following \cite{linearizedRS} one can construct linearized Reed-Solomon codes with $ \ell = q-1 $. When $ m=n $ grows and $ d $ is constant, the value $ \ell = q-1 = \Theta(q) $ attains asymptotically the general bound in Theorem \ref{th non existence MSRD}. 

\section{$\F_{q^m}$-linear codes}\label{sect:Fqmlinear}

Throughout this section, we restrict to the case $ m = m_1 = m_2 = \ldots = m_\ell $ and discuss some of the theory of $\mathbb{F}_{q^m}$-linear sum-rank metric codes in $\mathbb{F}_{q^m}^n$, where $ n = n_1 + n_2 + \cdots + n_\ell $. In this setting, many results from the previous sections may be simplified or even extended. In addition linearized Reed-Solomon codes, the main construction of MSRD codes, are linear over $\mathbb{F}_{q^m}$, whose size is subexponential in $n$. This is in contrast with Gabidulin codes in $\mathbb{F}_{q^m}^n$~\cite{Del, gabidulin, roth} and has the effect of making encoding and decoding more efficient for the same code length $n$.

Throughout the section, $\Cc \subseteq \mathbb{F}_{q^m}^n$ denotes an $\mathbb{F}_{q^m}$-linear code and $\dim(\Cc)$ is the dimension of $\Cc$ over $\F_{q^m}$.
We start by fixing an isomorphism between $ \mathbb{M} $ and $ \mathbb{F}_{q^m}^n $ in order to extend the sum-rank metric from $ \mathbb{M} $ to $ \mathbb{F}_{q^m}^n $, which allows us to consider $ \Cc $ as a sum-rank metric code. 

\begin{definition} \label{def matrix repr map}
Let $ \boldsymbol\alpha = (\alpha_1, \alpha_2, \ldots, \alpha_m) \in \mathbb{F}_{q^m}^m $ be an ordered basis of $ \mathbb{F}_{q^m} $ over $ \mathbb{F}_q $. For each positive integer $r$, we define the map $ M^r_{\boldsymbol\alpha} : \mathbb{F}_{q^m}^r \longrightarrow \mathbb{F}_q^{m \times r} $ as follows: For $ \mathbf{c} \in \mathbb{F}_{q^m}^r $, the matrix $ M_{\boldsymbol\alpha}^r(\mathbf{c}) $ is the only matrix in $ \mathbb{F}_q^{m \times r} $ such that $ \boldsymbol\alpha M_{\boldsymbol\alpha}^r(\mathbf{c}) = \mathbf{c} $. In other words, the entries of the $ i $th column of $ M_{\boldsymbol\alpha}^r(\mathbf{c}) $ are the coordinates of the $ i $th component of $ \mathbf{c} $ in the ordered basis $ \boldsymbol\alpha $.

For $ \mathbf{n} = (n_1, n_2, \ldots, n_\ell) $ and the ordered basis $ \boldsymbol\alpha \in \mathbb{F}_{q^m}^m $, define the map $ M^{\mathbf{n}}_{\boldsymbol\alpha} : \mathbb{F}_{q^m}^n \longrightarrow \mathbb{M} $ as follows: For $ \mathbf{c} = \left( \mathbf{c}^{(1)}, \mathbf{c}^{(2)}, \ldots, \mathbf{c}^{(\ell)} \right) \in \mathbb{F}_{q^m}^n $, where $ \mathbf{c}^{(i)} \in \mathbb{F}_{q^m}^{n_i} $ for $ i \in [\ell] $, we define 
$$ M^{\mathbf{n}}_{\boldsymbol\alpha}(\mathbf{c}) = \left( M^{n_1}_{\boldsymbol\alpha} \left( \mathbf{c}^{(1)} \right) , M^{n_2}_{\boldsymbol\alpha} \left( \mathbf{c}^{(2)} \right), \ldots, M^{n_\ell}_{\boldsymbol\alpha} \left( \mathbf{c}^{(\ell)} \right) \right). $$
\end{definition}

Clearly $ M^{\mathbf{n}}_{\boldsymbol\alpha} $ is an $ \mathbb{F}_q $-linear isomorphism between the vector spaces $ \mathbb{F}_{q^m}^n $ and $ \mathbb{M} $. In $ \mathbb{F}_{q^m}^n $, however, we can consider $ \mathbb{F}_{q^m} $-linear spaces in addition to $ \mathbb{F}_{q} $-linear spaces. Moreover, the isomorphism $ M^{\mathbf{n}}_{\boldsymbol\alpha} $ allows us to transfer the notion of sum-rank weight and distance from $\MM$ to $ \mathbb{F}_{q^m}^n $. Notice however that, in order to set up the isomorphism $ M^{\mathbf{n}}_{\boldsymbol\alpha} $ starting from $ \mathbb{F}_{q^m}^n $, one must fix a partition $ n = n_1 + n_2 + \cdots + n_\ell $ and a subfield $ \mathbb{F}_q $ (or equivalently, $ m $). 

Under this isomorphism, duality with respect to the usual inner product in $ \mathbb{F}_{q^m}^n $ corresponds to trace duality in $ \mathbb{M} $ according to Definition \ref{defn:dual}. The following is a straightforward extension of~\cite[Theorem 21]{Rav16a}.

\begin{proposition}
Let $ \Cc \subseteq \mathbb{F}_{q^m}^n $ be $ \mathbb{F}_{q^m} $-linear. Then
$$ M^{\mathbf{n}}_{\boldsymbol\alpha} \left( \Cc^\perp \right) = M^{\mathbf{n}}_{\boldsymbol\alpha} (\Cc)^\perp, $$
where $ \Cc^\perp = \{ \mathbf{d} \in \mathbb{F}_{q^m}^n \mid \mathbf{c}\mathbf{d}^\intercal = 0 \} $ is the dual of $ \Cc $ in $ \mathbb{F}_{q^m}^n $ with respect to the usual inner product.
\end{proposition}

Let $ \Cc \subseteq \mathbb{F}_{q^m}^n $ be an $ \mathbb{F}_{q^m} $-linear code of dimension $k$. For $r=1$, the Singleton Bound from Theorem \ref{singletonbound} takes the familiar form 
\begin{equation}\label{eq singleton bound vector form}
d(\Cc) \leq n - k + 1 . 
\end{equation}
This allows us to define MSRD codes as the codes that meet (\ref{eq singleton bound vector form}).
Notice that the bound also holds if $ n_i > m $ for some $ i \in [\ell] $, although when $ n_i > m $ for all $ i \in [\ell] $, there cannot be any codes attaining (\ref{eq singleton bound vector form}) by \cite[Proposition 1.8]{fnt}. Furthermore, all known $ \mathbb{F}_{q^m} $-linear MSRD codes have $ n_i \leq m $ for all $ i \in [\ell] $. For this reason, in the sequel we assume that $ n_i \leq m $ for all $ i \in [\ell] $.

For $ \mathbb{F}_{q^m} $-linear codes, duality preserves the MSRD property.

\begin{theorem}[{\cite[Theorem 5]{gsrws}}] \label{th dual of msrd is msrd}
An $ \mathbb{F}_{q^m} $-linear code $ \Cc \subseteq \mathbb{F}_{q^m}^n $ is MSRD if and only if its dual $ \Cc^\perp $ is MSRD.
\end{theorem}

In the case of $ \mathbb{F}_{q^m} $-linear isometries, one can refine Theorem~\ref{thm:isom} as follows.

\begin{theorem}[{\cite[Theorem 2]{sr-hamming}}]
A map $ \varphi : \mathbb{F}_{q^m}^n \longrightarrow \mathbb{F}_{q^m}^n $ is an $ \mathbb{F}_{q^m} $-linear isometry if and only if there exists a permutation $ \sigma : [\ell] \longrightarrow [\ell] $ with the property that $ \sigma(i) = j $ implies $ n_i = n_j $, and there exist scalars $ \beta_i \in \mathbb{F}_{q^m}^* $ and invertible matrices $ A_i \in \mathbb{F}_q^{n_i \times n_i} $, for $ i \in [\ell]$, such that
$$ \varphi(\mathbf{c}) = \left( \beta_1 \mathbf{c}^{(\sigma(1))} A_1, \beta_2 \mathbf{c}^{(\sigma(2))} A_2, \ldots, \beta_\ell \mathbf{c}^{(\sigma(\ell))} A_\ell \right), $$
for all $ \mathbf{c} = \left( \mathbf{c}^{(1)}, \mathbf{c}^{(2)}, \ldots, \mathbf{c}^{(\ell)} \right) \in \mathbb{F}_{q^m}^n $, where $ \mathbf{c}^{(i)} \in \mathbb{F}_{q^m}^{n_i} $ for $ i \in [\ell] $.
\end{theorem}

Let $\A\subseteq\MM$ be an optimal anticode.
Notice that $(M^{\mathbf{n}}_{\boldsymbol\alpha})^{-1}(\A)\subseteq\F_{q^m}^n$ may not be an $\F_{q^m}$-linear optimal anticode. In fact, $(M^{\mathbf{n}}_{\boldsymbol\alpha})^{-1}(\A)$ is $\F_{q^m}$-linear only if $\A$ is a row-support space. 
In the notation of~\cite{gsrws}, let $ \boldsymbol{\mathcal{L}} = (\mathcal{L}_1, \mathcal{L}_2, \ldots , \mathcal{L}_\ell) \in \prod_{i=1}^\ell \mathcal{P}(\mathbb{F}_q^{n_i}) $, where $ \mathcal{P}(\mathcal{V}) $ denotes the collection of subspaces of a vector space $ \mathcal{V} $. Define $ \mathcal{V}_{\boldsymbol{\mathcal{L}}} \subseteq \mathbb{F}_{q^m}^n $ as the $ \mathbb{F}_{q^m} $-linear subspace of vectors of the form $ \left( \mathbf{c}^{(1)}, \mathbf{c}^{(2)}, \ldots, \mathbf{c}^{(\ell)} \right) $, where $ \mathbf{c}^{(i)} \in \mathbb{F}_{q^m}^{n_i} $ is such that the row space of $ M^{\mathbf{n}}_{\boldsymbol\alpha} \left( \mathbf{c}^{(i)} \right) \in \mathbb{F}_q^{m \times n_i} $ is contained in $ \mathcal{L}_i $. Notice that $ \mathcal{V}_{\boldsymbol{\mathcal{L}}}=(M^{\mathbf{n}}_{\boldsymbol\alpha})^{-1}(\mathcal{V}_{M^{\mathbf{n}}_{\boldsymbol\alpha}(\boldsymbol{\mathcal{L}})})$, where $\mathcal{V}_{M^{\mathbf{n}}_{\boldsymbol\alpha}(\boldsymbol{\mathcal{L}})}$ is a row-support space as in Definition \ref{defn:rowsupp} . 

By restricting to $\F_{q^m}$-linear optimal anticodes, we obtain the following definition of generalized sum-rank weights of an $\mathbb{F}_{q^m}$-linear code.
\begin{definition}[{\cite[Definition 10]{gsrws}}]
Let $ \mathcal{C} \subseteq \mathbb{F}_{q^m}^n $ be an $\mathbb{F}_{q^m}$-linear code and let $ r \in [\dim(\Cc)] $.
The $r$-th generalized sum-rank weight of $\Cc$ is 
\begin{equation}
{\rm d}^\prime_r(\Cc) = \min \left\lbrace \sum_{i=1}^\ell \dim(\mathcal{L}_i) : \boldsymbol{\mathcal{L}} \in \prod_{i=1}^\ell \mathcal{P}_L(\mathbb{F}_q^{n_i}), \dim(\Cc \cap \mathcal{V}_{\boldsymbol{\mathcal{L}}}) \geq r \right\rbrace.
\label{eq def gsrw vector form}
\end{equation}
\end{definition}
    
For these generalized sum-rank weights, we have a simplified version of monotonicity and of the Singleton Bound:
$$ {\rm d}^\prime_r(\Cc) < {\rm d}^\prime_{r+1}(\Cc) \leq n \quad \textrm{and} \quad {\rm d}^\prime_r(\Cc) \leq n - k + r, $$
for $ r \in [\dim(\Cc)-1] $. In fact, most bounds valid for generalized Hamming weights are valid for generalized sum-rank weights as above, as shown in~\cite[Theorem 4]{gsrws}. In addition, Wei duality takes a simple form, which is reminiscent of Wei's original result~\cite[Theorem 3]{Wei}.

\begin{theorem}[{\cite[Theorem 2]{gsrws}}]
Let $ \Cc \subseteq \mathbb{F}_{q^m}^n $ be an $ \mathbb{F}_{q^m} $-linear code of $\dim(\Cc)=k$. For $ r \in [k] $ and $ s \in [n-k] $, let $ d_r = {\rm d}^\prime_r(\Cc) $ and $ d^\perp_s = {\rm d}^\prime_s(\Cc^\perp) $. Then
$$ [n] = \{ d_1, d_2, \ldots, d_k \} \cup \{ n+1-d_1^\perp, n+1-d_2^\perp, \ldots, n+1-d_{n-k}^\perp \}. $$
\end{theorem}

An interesting property of $ \mathbb{F}_{q^m} $-linear codes in $ \mathbb{F}_{q^m}^n $ is the existence of generator and parity-check matrices, which are matrices with entries in $\F_{q^m}$. By looking at the generator and parity-check matrices, we may characterize $ \mathbb{F}_{q^m} $-linear MSRD codes as follows. This characterization is useful in some applications, such as constructing PMDS codes~\cite{universal-lrc, cai}. The result generalizes the characterizations of MDS codes~\cite{MWSI} and MRD codes~\cite{gabidulin} based on generator and parity-check matrices. Here $ {\rm diag}(A_1, A_2, \ldots , A_\ell) $ denotes the block-diagonal matrix which has the blocks $ A_1, A_2, \ldots, A_\ell $ on the diagonal and zeros elsewhere.

\begin{theorem} [{\cite[Theorem 1.6]{fnt}}] \label{th gen and par matrices of MSRD}
Let $ \mathcal{C} \subseteq \mathbb{F}_{q^m}^n $ be an $ \mathbb{F}_{q^m} $-linear code of $\dim(\Cc)=k$. Let $ G \in \mathbb{F}_{q^m}^{k \times n} $ and $ H \in \mathbb{F}_{q^m}^{(n-k) \times n} $ be a generator matrix and a parity-check matrix, of $ \mathcal{C} $, respectively. The following are equivalent:
\begin{enumerate}
\item 
$ \mathcal{C} $ is MSRD.
\item 
Let $ A_i \in {\rm GL}_{n_i}(\mathbb{F}_q) $, for $ i \in [\ell] $.
Then every $ k \times k $ submatrix of $ G \cdot {\rm diag}(A_1, A_2, \ldots , A_\ell) \in \mathbb{F}_{q^m}^{k \times n} $ is invertible. 
\item 
Let $ A_i \in {\rm GL}_{n_i}(\mathbb{F}_q) $, for $ i \in [\ell] $.
Then every $ (n-k) \times (n-k) $ submatrix of $ H\,\cdot\, {\rm diag}(A_1, A_2, \ldots , A_\ell) \in \mathbb{F}_{q^m}^{(n-k) \times n} $ is invertible.
\end{enumerate}
\end{theorem}

When the generator matrix is in systematic form, we have the following stronger result, which generalizes the corresponding characterizations of MDS codes~\cite[Chapter 8, Theorem 8]{MWSI} and MRD codes~\cite[Theorem 3.18]{neri-systematic}.

\begin{theorem}[{\cite[Theorem 5]{systematic}}]\label{th systematic MSRD}
Let $ \mathcal{C} \subseteq \mathbb{F}_{q^m}^n $ be an $ \mathbb{F}_{q^m} $-linear code. Let $ k = \dim(\Cc) $ and consider an integer partition $ k = k_1 + k_2 + \cdots + k_\ell $, where $
0 \leq k_i \leq n_i $, for $ i \in [\ell] $. Let
$$ G = (J_1|P_1|J_2|P_2|\ldots|J_\ell|P_\ell) \in \mathbb{F}_{q^m}^{k \times n} $$
be a systematic generator matrix of $ \Cc $, where $ J_i \in \mathbb{F}_{q^m}^{k \times k_i} $, $ P_i \in \mathbb{F}_{q^m}^{k \times (n_i - k_i)} $, and 
$$ I_k = (J_1|J_2|\ldots|J_\ell) \in \mathbb{F}_{q^m}^{k \times k}$$
is the identity matrix. Let $ P = (P_1, P_2, \ldots, P_\ell) \in \mathbb{F}_{q^m}^{k \times (n-k)} $. Then $ \mathcal{C} $ is MSRD if and only if the square submatrices of all sizes of
$$ B P A + C \in \mathbb{F}_{q^m}^{k \times (n-k)} $$
are invertible, for all $ C_i \in \mathbb{F}_q^{k_i \times (n_i-k_i)} $, $ B_i \in {\rm GL}_{k_i}(\mathbb{F}_q) $, and $ A_i \in {\rm GL}_{n_i - k_i}(\mathbb{F}_q) $, for
$ i \in [\ell] $, where
\begin{equation}
\begin{split}
A & = {\rm diag}(A_1,A_2,\ldots,A_\ell) \in \mathbb{F}_q^{(n-k) \times (n-k)}, \\
B & = {\rm diag}(B_1,B_2,\ldots,B_\ell) \in \mathbb{F}_q^{k \times k}, \\
C & = {\rm diag}(C_1,C_2,\ldots,C_\ell) \in \mathbb{F}_q^{k \times (n-k)}.
\end{split}
\label{eq systematic gen matrix A,B,C}
\end{equation}
\end{theorem}

Finally, we discuss known families of $ \mathbb{F}_{q^m} $-linear MSRD codes. There exist several constructions of $ \mathbb{F}_{q^m} $-linear codes in $ \mathbb{F}_{q^m}^n $ for different length partitions $ n = n_1 + n_2 + \cdots + n_\ell $, see e.g.~\cite{linearizedRS, generalMSRD, twisted}. For square matrices, that is $ m = n_1=\ldots=n_\ell $, the only known construction is that of linearized Reed-Solomon codes~\cite{linearizedRS} and their twisted generalization \cite{twisted}. Linearized Reed-Solomon codes exist for any length partition $ n = n_1 + n_2 + \cdots + n_\ell $ such that $ n_i \leq m $, for $ i \in [\ell] $. Their main parameter drawback is that $ q > \ell $ is required.  

Even though one may define linearized Reed-Solomon codes over more general fields~\cite{linearizedRS}, an explicit definition for finite fields is as follows. First, for a positive integer $ k $, an element $ a \in \mathbb{F}_{q^m}^* $ and a vector $ \boldsymbol\beta = (\beta_1, \beta_2, \ldots, \beta_r) \in \mathbb{F}_{q^m}^r $, we define the matrix
\begin{equation}
M_k(a, \boldsymbol\beta) = \left( \begin{array}{cccc}
\beta_1 & \beta_2 & \ldots & \beta_r \\
\beta_1^q a & \beta_2^q a & \ldots & \beta_r^q a \\
\beta_1^{q^2} a^{q+1} & \beta_2^{q^2} a^{q+1} & \ldots & \beta_r^{q^2} a^{q+1} \\
\vdots & \vdots & \ddots & \vdots \\
\beta_1^{q^{k-1}} a^{\frac{q^{k-1}-1}{q-1}} & \beta_2^{q^{k-1}} a^{\frac{q^{k-1}-1}{q-1}} & \ldots & \beta_r^{q^{k-1}} a^{\frac{q^{k-1}-1}{q-1}}
\end{array} \right) \in \mathbb{F}_{q^m}^{k \times r}.
\label{eq def moore}
\end{equation}

\begin{definition} \label{def lrs}
Let $ a_1, a_2, \ldots, a_\ell \in \mathbb{F}_{q^m}^* $ such that $ N_{\mathbb{F}_{q^m}/\mathbb{F}_q}(a_i) \neq N_{\mathbb{F}_{q^m}/\mathbb{F}_q}(a_j) $ if $ i \neq j $, where $ N_{\mathbb{F}_{q^m}/\mathbb{F}_q} $ is the norm relative to the extension $ \mathbb{F}_{q^m}/\mathbb{F}_q $. Let $ \beta_1, \beta_2, \ldots, \beta_m \in \mathbb{F}_{q^m} $ be a basis of $ \mathbb{F}_{q^m} $ over $ \mathbb{F}_q $. Let $ k \in [n] $, $ n = n_1 + n_2 + \cdots + n_\ell $, and $ n_i \leq m $ for $ i \in [\ell] $. The linearized Reed-Solomon code of dimension $ k $ is the code $ \mathcal{C}_{LRS}^{n,k}(\mathbf{a}, \boldsymbol\beta) \subseteq \mathbb{F}_{q^m}^n $ with generator matrix
$$ M_k = (M_k(a_1, \boldsymbol\beta_1) | M_k(a_2, \boldsymbol\beta_2) | \ldots | M_k(a_\ell, \boldsymbol\beta_\ell)) \in \mathbb{F}_{q^m}^{k \times n}, $$
where $ \boldsymbol\beta_i = (\beta_1, \beta_2, \ldots, \beta_{n_i}) $ for all $ i \in [\ell] $.
\end{definition}

From the shape of the generator matrix (see also (\ref{eq def moore})), it is easy to see that linearized Reed-Solomon codes recover classical generalized Reed-Solomon codes when $ m=1 $ (thus $ n_1 = \ldots = n_\ell = 1 $), and Gabidulin codes \cite{gabidulin} when $ \ell = 1 $ and $ a_1 = 1 $. These are also the cases in which the sum-rank metric specializes to the Hamming and rank metric, respectively.

In order to prove that linearized Reed-Solomon codes are MSRD, we need a result concerning roots of skew polynomials. In our setting, a skew polynomial \cite{ore} is a polynomial expression with coefficients on the left, 
$$ f = f_0 + f_1 x + \cdots + f_d x^d, $$
where $ d $ is a non-negative integer and $ f_0, f_1, \ldots, f_d \in \mathbb{F}_{q^m} $. The set of all skew polynomials of this form is denoted by $ \mathbb{F}_{q^m}[x;\sigma] $, where $ \sigma : \mathbb{F}_{q^m} \longrightarrow \mathbb{F}_{q^m} $ is the field automorfism given by $ \sigma(a) = a^q $, for $ a \in \mathbb{F}_{q^m} $. The set $ \mathbb{F}_{q^m}[x;\sigma] $ is a (generally non-commutative) ring with classical addition and multiplication by scalars in $ \mathbb{F}_{q^m} $ on the left, where $ x^i x^j = x^{i+j} $ for $ i,j \in \mathbb{N} $ and
$$ xa = \sigma(a)x $$
for all $ a \in \mathbb{F}_{q^m} $. Linearized Reed-Solomon codes may also be described using evaluation of skew polynomials. For $ f = f_0 + f_1 x + \cdots + f_d x^d \in \mathbb{F}_{q^m}[x;\sigma] $ and $ a, \beta \in \mathbb{F}_{q^m} $, we define
\begin{equation*}
\begin{split}
f_a(\beta) & = f_0 \beta N_0(a) + f_1 \beta^q N_1(a) + f_2 \beta^{q^2} N_2(a) + \cdots + f_d \beta^{q^d} N_d(a) \\
& = f_0 \beta + f_1 \beta^q a + f_2 \beta^{q^2} a^{\frac{q^2-1}{q-1}} + \cdots + f_d \beta^{q^d} a^{\frac{q^d-1}{q-1}} ,
\end{split}
\end{equation*}
where $ N_i(a) = a^{q^{i-1}} a^{q^{i-2}} \cdots a^q a = a^{\frac{q^i-1}{q-1}} $ is called the $ i $th truncated norm of $ a $ (since $ N_m(a) = N_{\mathbb{F}_{q^m}/\mathbb{F}_q}(a) $). Notice that the map $ f_a : \mathbb{F}_{q^m} \longrightarrow \mathbb{F}_{q^m} $ is $ \mathbb{F}_q $-linear, and the set $ \{ f_1 : f \in \mathbb{F}_{q^m}[x;\sigma] \} $ is the classical set of $ q $-linearized polynomials over $ \mathbb{F}_{q^m} $ \cite[Section 3.4]{lidl}. If $ M_k $ is the generator matrix of the linearized Reed-Solomon code $ \mathcal{C}_{LRS}^{n,k}(\mathbf{a}, \boldsymbol\beta) \subseteq \mathbb{F}_{q^m}^n $ from Definition \ref{def lrs}, then its codewords are evaluations
$$ (f_{a_1}(\beta_1), \ldots, f_{a_1}(\beta_{n_1}), \ldots, f_{a_\ell}(\beta_1), \ldots, f_{a_\ell}(\beta_{n_\ell})) = (f_0, f_1, \ldots, f_{k-1}) \cdot M_k, $$
for skew polynomials $ f = f_0 + f_1 x + \cdots + f_{k-1}x^{k-1} \in \mathbb{F}_{q^m}[x;\sigma] $ of degree less than $ k $. The next auxiliary lemma is a combination of  \cite[Theorem 4.5]{lam-leroy}, \cite[Theorem 23]{lam} and \cite[Lemma 24]{linearizedRS} in the special case of finite fields (see also \cite[Theorem 2.1]{leroy-noncommutative} for a more general result). The result is of independent interest, and it recovers the following well-known statements:
The number of zeros of a polynomial is bounded from above by its degree and the dimension of the zero-locus of a $q$-linearized polynomial is bounded from above by the logarithm in base $q$ of its degree.

\begin{lemma} \label{lemma roots}
Let $ a_1, a_2, \ldots, a_\ell \in \mathbb{F}_{q^m}^* $ such that $ N_{\mathbb{F}_{q^m}/\mathbb{F}_q}(a_i) \neq N_{\mathbb{F}_{q^m}/\mathbb{F}_q}(a_j) $ if $ i \neq j $. For any non-zero $ f \in \mathbb{F}_{q^m}[x;\sigma] $, we have
$$ \sum_{i=1}^\ell \dim_{\mathbb{F}_q}(\ker(f_{a_i})) \leq \deg(f). $$
\end{lemma}

We may now prove that linearized Reed-Solomon codes are MSRD.

\begin{theorem}[{\cite[Theorem 4]{linearizedRS}}]
The linearized Reed-Solomon code $ \mathcal{C}_{LRS}^{n,k}(\mathbf{a}, \boldsymbol\beta) \subseteq \mathbb{F}_{q^m}^n $ is MSRD for the partition $ n = n_1 + n_2 + \cdots + n_\ell $ and the subfield $ \mathbb{F}_q $.
\end{theorem}
\begin{proof}
Let $ M_k $ be the generator matrix of $ \mathcal{C}_{LRS}^{n,k}(\mathbf{a}, \boldsymbol\beta) $ as in Definition \ref{def lrs}. Notice that, by the linearity of the map $ \sigma $, we have
$$ M^\prime_k := M_k \cdot {\rm diag}(A_1, A_2, \ldots, A_\ell) = (M_k(a_1, \boldsymbol\beta_1 A_1) | M_k(a_2, \boldsymbol\beta_2 A_2) | \ldots | M_k(a_\ell, \boldsymbol\beta_\ell A_\ell)), $$
for all $ A_i \in {\rm GL}_{n_i}(\mathbb{F}_q) $, where the components $ \beta^\prime_{i,1}, \beta^\prime_{i,2}, \ldots, \beta^\prime_{i,n_i} $ of $ \boldsymbol\beta^\prime_i = \boldsymbol\beta_i A_i \in \mathbb{F}_{q^m}^{n_i} $ are $ \mathbb{F}_q $-linearly independent, for $ i \in [\ell] $. By Theorem \ref{th gen and par matrices of MSRD}, we only need to prove that the code with generator matrix $ M^\prime_k $ is MDS. As before, the codewords of such a code are of the form
$$ (f_{a_1}(\beta^\prime_{1,1}), \ldots, f_{a_1}(\beta^\prime_{1,n_1}), \ldots, f_{a_\ell}(\beta^\prime_{\ell,1}), \ldots, f_{a_\ell}(\beta^\prime_{\ell,n_\ell})) = (f_0, f_1, \ldots, f_{k-1}) \cdot M^\prime_k, $$
for some $ f = f_0 + f_1 x + \cdots + f_{k-1} x^{k-1} \in \mathbb{F}_{q^m}[x;\sigma] $. If the previous codeword has Hamming weight at most $ n-k $, then we deduce that
$$ \sum_{i=1}^\ell \dim_{\mathbb{F}_q}(\ker(f_{a_i})) \geq k, $$
since $ \beta^\prime_{i,1}, \beta^\prime_{i,2}, \ldots, \beta^\prime_{i,n_i} $ are $ \mathbb{F}_q $-linearly independent. Since $ \deg(f) < k $, we conclude that $ f = 0 $ by Lemma \ref{lemma roots} and we are done.
\end{proof}

Notice that, since $ N_{\mathbb{F}_{q^m}/\mathbb{F}_q}(\mathbb{F}_{q^m}^*) = \mathbb{F}_q^* $, we may choose up to $q-1$ elements $ a_1, a_2, \ldots, a_{q-1} $ with distinct norms in Definition \ref{def lrs}. In other words, the restrictions on the parameters of a linearized Reed-Solomon code are $ \ell \leq q-1 $ and $ n_i \leq m $ for $ i \in [\ell] $, where equalities can be attained. Hence, they are linear over a field that can be chosen of size 
$$ q^m = (\ell+1)^{\max\{ n_1, n_2, \ldots, n_\ell \}} . $$  
Decoding algorithms for linearized Reed-Solomon codes over $ \mathbb{F}_{q^m} $ with cubic complexity~\cite{boucher-skew}, quadratic complexity~\cite{secure-multishot}, and subquadratic complexity~\cite{sven-sum-rank} have been recently developed. Other decoders of linearized Reed-Solomon codes and variants include \cite{interleaved, hormann, hormannefficient}. An algorithm with exponential complexity for decoding an arbitrary linear code in the sum-rank metric appears in \cite{generic}. The complexity of list decoding linearized Reed-Solomon codes was studied in~\cite{list-decoding}. 

A family of MSRD codes generalizing linearized Reed-Solomon codes and their parameters was proposed in~\cite{generalMSRD}. In particular, they may have smaller field sizes relative to their parameters. The idea is to eliminate the assumption that $ \beta_1, \beta_2, \ldots, \beta_{n_i} $ are $ \mathbb{F}_q $-linearly independent. We have the following characterization.

\begin{theorem}[{\cite[Theorem 3.12]{generalMSRD}}] \label{th msrd code general characterization}
Let $ a_1, a_2, \ldots, a_t \in \mathbb{F}_{q^m}^* $ be such that $ N_{\mathbb{F}_{q^m}/\mathbb{F}_q}(a_i) \neq N_{\mathbb{F}_{q^m}/\mathbb{F}_q}(a_j) $ if $ i \neq j $. Let $ \mu $ and $ r $ be positive integers and let $ \boldsymbol\beta = (\beta_1, \beta_2, \ldots, \beta_{\mu r}) \in \mathbb{F}_{q^m}^{\mu r} $. For $ i \in [\mu] $, let $ \mathcal{H}_i \subseteq \mathbb{F}_{q^m} $ be the $ \mathbb{F}_q $-linear subspace generated by $ \beta_{(i-1)r+1}, \beta_{(i-1)r+2}, \ldots, \beta_{ir} $. Let $ \ell = t \mu $ and $ n = \ell r $. For $ k \in [n] $, the linear code with generator matrix
\begin{equation}
(M_k(a_1, \boldsymbol\beta) | M_k(a_2, \boldsymbol\beta) | \ldots | M_k(a_t, \boldsymbol\beta)) \in \mathbb{F}_{q^m}^{k \times n}
\label{eq gen matrix msrd code general}
\end{equation}
is MSRD for the partition $ n = r + r + \cdots + r $ ($ \ell $ times) and the subfield $ \mathbb{F}_q $ if and only if the following two conditions hold:
\begin{enumerate}
\item
$ \dim_{\mathbb{F}_q}(\mathcal{H}_i) = r $, and
\item
$ \mathcal{H}_i \cap \left( \sum_{j \in \Gamma} \mathcal{H}_j \right) = \{ 0 \} $ for any $ \Gamma \subseteq [\mu] $ such that $ i \notin \Gamma $ and $ |\Gamma| \leq \min \{ k,\mu \} -1 $.
\end{enumerate}
\end{theorem}

Notice that the same conditions can be applied to the parity-check matrix of a code as in (\ref{eq gen matrix msrd code general}), since the dual of a linear MSRD code is also MSRD by Theorem~\ref{th dual of msrd is msrd}. This would entail replacing $ k $ by $ n-k $ in the second condition.

Several constructions of families of subspaces satisfying Conditions 1 and 2 were given in~\cite{generalMSRD}, see \cite[Section 4.1]{generalMSRD}. Each construction yields an $ \mathbb{F}_{q^m} $-linear MSRD code as in Theorem~\ref{th msrd code general characterization} with a field size $ q^m $ that is the smallest among known $ \mathbb{F}_{q^m} $-linear MSRD codes for their corresponding parameters, see Table \ref{table MSRD}. The case of linearized Reed-Solomon copdes is recovered by choosing $ \mu = 1 $.
Currently, there are no known efficient decoding algorithms that apply to the MSRD codes from Theorem~\ref{th msrd code general characterization}. Similarly, explicit descriptions of their duals are not yet known.

\begin{table} [!t]
\centering
\resizebox{\textwidth}{!}{\begin{tabular}{c||c|c|c}
\hline
&&\\[-0.8em]
Code $ \mathcal{C}_{\boldsymbol\gamma} $ & $ q $, $ r $, $ h $ & $ \ell = t \mu = (q-1) \mu $ & Field size $ q^m $ \\[0.3em]
\hline\hline
&&\\[-0.8em]
Trivial $ \mathcal{C}_{\boldsymbol\gamma} = \{ 0 \} $ & Any & $ q-1 $ & $ q^r = (\ell +1)^r $, $ m = r $ \\[0.3em]
\hline 
&&\\[-0.8em]
MDS & Any & $ (q-1) \left( q^r + 1 \right) $ & $ \left( \frac{\ell}{q-1} -1 \right)^{\min \left\lbrace h, \frac{\ell}{q-1} \right\rbrace } $ \\[0.3em]
\hline 
&&\\[-0.8em]
Hamming, $ \rho \in \mathbb{Z}^+ $ & $ h = 2 $ & $ (q-1) \cdot \frac{q^{r \rho} - 1}{q^r - 1} $ & $ q^{r \rho} = \frac{q^r - 1}{q-1} \cdot \ell + 1 $ \\[0.3em]
\hline 
&&\\[-0.8em]
BCH, $ s \in \mathbb{Z}^+ $ & Any & $ (q-1) \left( q^{rs} - 1 \right) $ & $ \leq q^r \cdot \left( \frac{\ell}{q-1} + 1 \right) ^{ \left\lceil \frac{q^r - 1}{q^r} (h-1) \right\rceil} $ \\[0.3em]
\hline 
&&\\[-0.8em]
Hermitian curves & $ q^r = p^{2s} $ & $ (q-1) q^{ \frac{3r}{2} } $ & $ \mu^{ \frac{1}{3} ( 2h + \mu^{2/3} - \mu^{1/3} ) } $, $ \mu = \frac{\ell}{q-1} $ \\[0.3em]
\hline 
&&\\[-0.8em]
Suzuki curves & $ q^r = 2^{2s+1} $ & $ (q-1) q^{2r} $ & $ \leq \mu^{ \frac{1}{2} \left( h + \mu^{3/4} - \mu^{1/4} \right) } $, $ \mu = \frac{\ell}{q-1} $ \\[0.3em]
\hline 
&&\\[-0.8em]
Optimal curves, $ i \in \mathbb{Z}^+ $ & $ q^r = p^{2s} $ & $ (q-1) \left( q^{ \frac{r}{2} } -1 \right) q^{ \frac{i r}{2} } $ & $ \leq \left( \frac{\mu_i}{ q^{ \frac{r}{2} } - 1 } \right)^{ \frac{2}{i} \left( h_i + \frac{\mu_i}{ q^{ \frac{r}{2} } - 1 } \right) } $, $ \mu_i = \frac{\ell_i}{q-1} $ \\[0.3em]
\hline 
\end{tabular}}
\caption{Table summarizing the achievable code parameters of the $ \mathbb{F}_{q^m} $-linear MSRD codes obtained in \cite{generalMSRD} based on Theorem \ref{th msrd code general characterization}. The sequence $ \beta_1, \beta_2, \ldots, \beta_{r \mu} $ is obtained using a Hamming-metric $ \mathbb{F}_{q^r} $-linear code $ \mathcal{C}_{\boldsymbol\gamma} \subseteq \mathbb{F}_{q^r}^\mu $, and each row corresponds to a different choice of $ \mathcal{C}_{\boldsymbol\gamma} $. Here, we consider the matrix (\ref{eq gen matrix msrd code general}) with $ h $ rows as a parity-check matrix of the $ \mathbb{F}_{q^m} $-linear MSRD code, whose dimension is therefore $ k = n-h $. The codes in the last three rows are Algebraic-Geometry codes obtained from the curves indicated. }
\label{table MSRD}
\end{table}

We conclude by briefly describing twisted linearized Reed-Solomon codes \cite{twisted}. They are a generalization of linearized Reed-Solomon codes and include $ \mathbb{F}_q $-linear MSRD codes with the same parameters as linearized Reed-Solomon codes, but which are not isometric to them. Notice that only some of the codes from \cite{twisted} are $ \mathbb{F}_{q^m} $-linear. Nevertheless, we decided to discuss them in this section due to their connection to linearized Reed-Solomon codes. The following is \cite[Definition 6.2]{twisted} for finite fields.

\begin{definition}
Let $ a_1, a_2, \ldots, a_\ell \in \mathbb{F}_{q^m}^* $ such that $ N_{\mathbb{F}_{q^m}/\mathbb{F}_q}(a_i) \neq N_{\mathbb{F}_{q^m}/\mathbb{F}_q}(a_j) $ if $ i \neq j $. Let $ \beta_1, \beta_2, \ldots, \beta_m \in \mathbb{F}_{q^m} $ be a basis of $ \mathbb{F}_{q^m} $ over $ \mathbb{F}_q $. Let $ k \in [n] $, and $ n_i = m $ for $ i \in [\ell] $, hence $ n = \ell m $. Let $ \eta \in \mathbb{F}_{q^m} $ such that $ (-1)^{kn} N_{\mathbb{F}_{q^m}/\mathbb{F}_q}(\eta) $ is not in the multiplicative group generated by $ \{ N_{\mathbb{F}_{q^m}/\mathbb{F}_q}(a_i) : 1 \leq i \leq \ell \} $. 

The twisted linearized Reed-Solomon code of dimension $ k $ is the $ \mathbb{F}_q $-linear code $ \mathcal{C}_{tLRS}^{n,k,\eta,h}(\mathbf{a}, \boldsymbol\beta) \subseteq \mathbb{F}_{q^m}^n $ whose codewords are of the form
$$ (f_0, \ldots, f_{k-1}) \cdot (M_k(a_1, \boldsymbol\beta) | M_k(a_2, \boldsymbol\beta) | \ldots | M_k(a_\ell, \boldsymbol\beta)) + \eta f_0^{q^h} \left( \beta_1^{q^k} a_1^{\frac{q^k-1}{q-1}} , \beta_2^{q^k} a_1^{\frac{q^k-1}{q-1}} , \ldots , \beta_m^{q^k} a_\ell^{\frac{q^k-1}{q-1}} \right), $$
for any $ f_0, f_1, \ldots , f_{k-1} \in \mathbb{F}_{q^m} $, where $ \boldsymbol\beta = (\beta_1, \beta_2, \ldots, \beta_m) $.
\end{definition}

The last term added makes the code only $ \mathbb{F}_q $-linear and not $ \mathbb{F}_{q^m} $-linear in general, since the map $ f_0 \mapsto \eta f_0^{q^h} $ is not always $ \mathbb{F}_{q^m} $-linear. Linearized Reed-Solomon codes are obtained by letting $ \eta = 0 $. Notice that, in that case, $ (-1)^{kn} N_{\mathbb{F}_{q^m}/\mathbb{F}_q}(\eta) = 0 $ does not belong to any multiplicative subgroup of $ \mathbb{F}_q^* $. Similarly to the case of linearized Reed-Solomon codes, twisted linearized Reed-Solomon codes can also be seen as evaluations 
$$ (f_{a_1}(\beta_1), \ldots, f_{a_1}(\beta_m), \ldots, f_{a_\ell}(\beta_1), \ldots, f_{a_\ell}(\beta_m)) \in \mathbb{F}_{q^m}^{\ell m} $$
of skew polynomials of the form
$$ f = f_0 + f_1 x + \cdots + f_{k-1}x^{k-1} + \eta f_0^{q^h} x^k \in \mathbb{F}_{q^m}[x;\sigma]. $$

We observe that, when $ \ell = 1 $ and $ a_1 = 1 $, twisted linearized Reed-Solomon codes recover twisted Gabidulin codes \cite{sheekey}, and when $ m = 1 $ (thus $ n_1 = \ldots = n_\ell = 1 $) twisted linearized Reed-Solomon codes recover twisted Reed-Solomon codes \cite{beelen} as a special case. There exist further generalizations of twisted Gabidulin codes. Such generalizations can also be extended to twisted linearized Reed-Solomon codes, see \cite[Remark 6.7]{twisted} and \cite[Section 7.1]{twisted}.

We also have the following.

\begin{theorem}[{\cite[Theorem 6.3]{twisted}}]
The twisted linearized Reed-Solomon code $ \mathcal{C}_{tLRS}^{n,k,\eta,h}(\mathbf{a}, \boldsymbol\beta) \subseteq \mathbb{F}_{q^m}^n $ is MSRD for the partition $ n =\ell m= m + m + \cdots + m $ and the subfield $ \mathbb{F}_q $.
\end{theorem}

Other properties of interest are discussed in \cite{twisted}, including attainable parameters \cite[Sec. 6.2]{twisted} and duality \cite[Sec. 6.3]{twisted}. Their decoding has not yet been studied, to the best of our knowledge. We conclude by discussing the possible values of the parameters of twisted linearized Reed-Solomon codes. Compared to linearized Reed-Solomon codes, we have the additional restriction that $ (-1)^{kn} N_{\mathbb{F}_{q^m}/\mathbb{F}_q}(\eta) $ must not lie in the multiplicative group generated by $ \{ N_{\mathbb{F}_{q^m}/\mathbb{F}_q}(a_i) : 1 \leq i \leq \ell \} $, which we denote by $ G $. 

We always have the restriction $ 1 \leq \ell \leq q-1 $. The case $ \ell = q-1 $ is equivalent to $ G = \mathbb{F}_q^* $ being the whole multiplicative group. In that case, the only possible choice is $ \eta = 0 $, hence we only obtain linearized Reed-Solomon codes. Non-trivial twisted linearized Reed-Solomon codes \cite[Proposition 6.8]{twisted} correspond to $ \eta \neq 0 $ and can be obtained whenever $ \ell \leq (q-1)/r $, where $ r $ is the smallest prime dividing $ q-1 $. Over infinite fields, one can obtain non-trivial twisted linearized Reed-Solomon codes for any choice of $ \ell $.


\end{document}